\documentclass[a4paper,fleqn]{scrartcl}

\usepackage[a4paper,bottom=3.5cm]{geometry} 

\usepackage[english]{babel}
\usepackage[utf8]{inputenc} 

\usepackage{graphicx, amsmath, amssymb, amsfonts, amsthm, tikz, subfig, rotating, xspace}
\usetikzlibrary{calc,decorations.pathmorphing,fit,backgrounds,positioning,arrows,shapes}

\usepackage{hyperref}
\usepackage[all]{hypcap}

\usepackage[numbers]{natbib}
\usepackage{tabularx}

\renewcommand{\o}[1]{\overline{#1}}

\usepackage{colonequals}

\usepackage{tabularx}

\usepackage{mathpazo}
\usepackage[euler-digits,T1]{eulervm}

\setkomafont{captionlabel}{\small\sffamily\bfseries}
\setkomafont{caption}{\small}

\usepackage{setspace}
\usepackage{algorithm}
\usepackage[noend]{algpseudocode}
\floatname{algorithm}{\textbf{\sffamily Algorithm}}

\usepackage{eqparbox,array}

\makeatletter
\gdef\reallynopagebreak{%
	\par\nopagebreak\@nobreaktrue%
}%
\makeatother

\newtheoremstyle{thmstyle}
				{10pt}
				{5pt}
				{\hangindent=0pt}
				{}
				{\bfseries \sffamily}
				{:\reallynopagebreak}
				{\newline}
				{}

\newtheoremstyle{remarkstyle}
				{10pt}
				{5pt}
				{\hangindent=0pt}
				{}
				{\bfseries \sffamily}
				{:\reallynopagebreak}
				{\newline}
				{}

\newtheorem{thmcounter}{}

\theoremstyle{thmstyle}
\newtheorem{theorem}[thmcounter]{Theorem}

\newtheorem{lemma}[thmcounter]{Lemma}

\theoremstyle{remarkstyle}
\newtheorem{ndefinition}[thmcounter]{Definition}
\newtheorem{nexample}[thmcounter]{Example}
\newtheorem{ncomment}[thmcounter]{Remark}
\newtheorem{nobservation}[thmcounter]{Observation}
\newenvironment{definition}{\begin{ndefinition}}{\mbox{}\hfill$\lhd$\end{ndefinition}}

\newenvironment{observation}{\begin{nobservation}}{\mbox{}\hfill$\lhd$\end{nobservation}}

\makeatletter
\renewenvironment{proof}[1][\proofname]{\par
  \pushQED{\qed}%
  \normalfont \topsep6\p@\@plus6\p@\relax
  \trivlist
  \item[\hskip\labelsep
        \bfseries \sffamily #1\@addpunct{:}]
        \mbox{}\noindent\ignorespaces
}{%
  \popQED\endtrivlist\@endpefalse
}
\makeatother

\global\long\def\NP{\mathcal{NP}}

\DeclareMathOperator{\sgn}{sgn}
\DeclareMathOperator{\SP}{SP}

\renewcommand{\o}[1]{%
	\overline{#1}%
}
\renewcommand{\u}[1]{%
	\underline{#1}%
}

\newcommand{\scaledtikz}[3][]{\begin{tikzpicture}\node[scale=#2] {\begin{tikzpicture}[#1] #3 \end{tikzpicture}};\end{tikzpicture}}

\newcommand{\Overset}[3][0pt]{%
	\ensuremath{\overset{\raise#1\hbox{\scriptsize\ensuremath{#2}}}{#3}}%
}

\newcommand{\citets}[1]{\citeauthor{#1}'s (\citeyear{#1})}

\def\OO{\mathcal{O}}
\def\OOT{\mathcal{\widetilde{O}}}

\definecolor{darkgreen}{rgb}{0,0.5,0}
\definecolor{halfgray}{gray}{0.55}
\tikzstyle{knoten}=[circle,draw=black,thin,fill=white,inner sep=0pt,minimum size=4.5mm]
\tikzstyle{knotenklein}=[circle,draw=black,thin,fill=white,inner sep=0pt,minimum size=2.5mm]
\tikzstyle{rot}=[draw=red,solid,thick]
\tikzstyle{blau}=[draw=blue,dashed,thick]
\tikzstyle{schwarz}=[draw=black,dotted,thick]
\tikzstyle{grau}=[draw=gray,dotted,thin]
\tikzstyle{gruen}=[draw=darkgreen,dashdotted,thick]
\tikzstyle{lila}=[draw=orange,dotted,thick]
\tikzset{dotted/.style={dash pattern=on 1pt off 2pt}}
\tikzset{dashed/.style={dash pattern=on 4pt off 3pt}}
\tikzset{dashdotted/.style={dash pattern=on 1pt off 2pt on 4pt off 2pt}}
\tikzset{dashdodotted/.style={dash pattern=on 1pt off 2pt on 1pt off 2pt on 4pt off 2pt}}

\setkomafont{disposition}{\normalcolor}

\title{A Generalized Approximation Framework for Fractional Network Flow and Packing Problems\thanks{This work was partially supported by the German Federal Ministry of Education and Research within the project ``SinOptiKom -- Cross-sectoral Optimization of Transformation Processes in Municipal Infrastructures in Rural Areas''.}}

\author{Michael Holzhauser \and Sven O. Krumke}

\date{\vspace{3mm}\normalsize University of Kaiserslautern, Department of Mathematics\vspace{-3mm}}

\begin{document}
	\maketitle

	\begin{abstract}
		We generalize the fractional packing framework of \citet{GargKoenemann} to the case of linear fractional packing problems over polyhedral cones. More precisely, we provide approximation algorithms for problems of the form~$\max\{c^T x : Ax \leq b, x \in C \}$, where the matrix~$A$ contains no negative entries and $C$ is a cone that is generated by a finite set~$S$ of non-negative vectors. While the cone is allowed to require an exponential-sized representation, we assume that we can access it via one of three types of oracles. For each of these oracles, we present positive results for the approximability of the packing problem. In contrast to other frameworks, the presented one allows the use of arbitrary linear objective functions and can be applied to a large class of packing problems without much effort. In particular, our framework instantly allows to derive fast and simple fully polynomial-time approximation algorithms (FPTASs) for a large set of network flow problems, such as budget-constrained versions of traditional network flows, multicommodity flows, or generalized flows. Some of these FPTASs represent the first ones of their kind, while others match existing results but offer a much simpler proof.
	\end{abstract}

	\section{Introduction}
	\label{sec:Introduction}

	In a fractional linear packing problem, one seeks to find a solution to the problem $\max\{c^Tx : Ax \leq b, x \geq 0 \}$, where the matrix~$A \in \mathbb{N}^{m \times n}_{\geq 0}$ contains no negative entries and, without loss of generality, the vectors~$c \in \mathbb{N}_{> 0}^n$ and $b \in \mathbb{N}_{> 0}^m$ have positive entries. Many problems can be formulated as packing problems, possibly the most intuitive being the fractional knapsack problem (cf. \citep{Knapsack}). More than this, many network flow problems can be seen as fractional packing problems if one allows exponential sized representations. For example, the traditional maximum flow problem can be seen as the problem of packing flows on $s$-$t$-paths without violating the capacities of the edges.

	A large number of authors presented approximation frameworks for such fractional packing problems, including \citet{PlotkinShmoysTardosFractionalPacking}, \citet{GrigoriadisKhachiyanBlockAngular,GrigoriadisKhachiyanParallelDecomposition}, \citet{YoungRandomizedRounding}, and \citet{BienstockIyengarFractionalPacking} (cf. \citep{BienstockLPApprox,GargKoenemann} for an overview of these results). One of the most powerful frameworks among these was developed by \citet{GargKoenemann}, who have shown that a $(1 - \varepsilon)$-approximate solution to a general fractional packing problem of the above form can be computed efficiently \emph{provided} we are able to ``handle'' the dual problem appropriately. More precisely, in the dual formulation~$\min\{b^T y: A^T y \geq c, y \geq 0\}$, we need to be able to determine a \emph{most violated dual constraint} efficiently: For some given infeasible solution~$y$ to the dual, we need to find a dual constraint\footnote{We use the notation~$B_{l\cdot}$ for a matrix~$B$ to denote the $l$-th row of the matrix.}~$(A^T)_{l\cdot} y \geq c_l$ that minimizes the value $\frac{(A^T)_{l\cdot} y}{c_l}$ among all dual constraints with a positive right-hand side value. This constraint reveals the largest degree of violation. Since the result may even hold if the number of variables in the primal formulation is of exponential size, the authors were able to provide efficient FPTASs for network flow problems such as multicommodity flow problems by using (exponential-sized) path-based formulations of the corresponding problems (cf. \citet{GargKoenemann}). \citet{FleischerMulticommodity} later showed that it suffices to determine an \emph{approximately most violated dual constraint}, whose fraction $\frac{(A^T)_{l\cdot} y}{c_l}$ may be up to a factor~$1 + \varepsilon$ away from the largest violation.

	In this paper, we generalize the result of \citet{GargKoenemann} to the case of packing problems over polyhedral cones. More precisely, we are interested in approximate solutions to problems of the form
	\begin{subequations}\label{eqn:OriginalProblem}%
	\begin{align}%
		\max\ & c^T x \\
		\text{s.t.}\ & Ax \leq b, \\
		& x \in C,
	\end{align}%
	\end{subequations}%
	where the matrix~$A \in \mathbb{N}^{m \times n}_{\geq 0}$ contains no negative numbers and where the cone~$C$ is finitely generated by a set~$S$ of non-negative vectors. While the vector $b \in \mathbb{N}^m_{> 0}$ is still assumed to contain positive entries (without loss of generality), we allow the entries of the vector~$c \in \mathbb{N}^n$ to have arbitrary signs. We will thereby combine a large set of well-known techniques such as the fractional packing framework of \citet{GargKoenemann}, the extension of \citet{FleischerMulticommodity}, the parametric search technique of \citet{MegiddoCombinatorialOptimization}, geometric-mean binary search due to \citet{HassinGeometricBinarySearch}, and transformation strategies for fractional objectives as described in \citet{LawlerCombinatorialOptimization}.

	The chosen formulation is motivated by the following observation: Network flow problems often allow a characterization by some kind of \emph{flow decomposition}, i.e., each feasible flow is representable by the sum of flows on much simpler structures, which we call \emph{basic components} in the following. Viewed from the other side, we can express each such feasible flow as a conic combination of flows on these basic components. Hence, if $C$ describes the cone that is generated by flows on basic components, we can express structural properties of each flow by the containment in the cone. What usually remains are packing constraints that bound the total flow on each edge, the flow that leaves some node, or the overall costs of the flow. As it will be shown, most common network flow problems can be modeled in such a way.

	One major advantage of the presented framework is that we do not assume the cone~$C$ or the set~$S$ that generates it to be given explicitly. Instead, we only assume that we have some kind of oracle access to the cone, which allows us to derive polynomial-time algorithms for problems that require cones with an exponential number of extreme rays. In addition to this benefit, our framework provides the following advantages: 
	\begin{itemize}
		\item The framework allows to derive fast and simple combinatorial FPTASs for a large class of packing problems and network flow problems. In many cases, we are even able to derive first strongly polynomial-time FPTASs.
		\item As our problem is formulated as a packing problem, the addition of further budget-constraints does not influence the applicability of the procedure, whereas such constraints usually make the design of exact algorithms significantly harder.
		\item The application of the framework only requires two ``ingredients'', namely the existence of some kind of flow decomposition theorem and a decent amount of control over the resulting basic components.
		\item As the framework is based on the approximation scheme of \citet{GargKoenemann} in its core, it works without considering some kind of residual network in case of network flow problems. As a consequence, the framework can be applied to problems that do not offer a concept of residual networks. One example is the maximum flow problem in generalized processing networks that will be discussed later. Moreover, it maintains properties of the underlying network such as cycle freeness or signs of costs.
		\item In contrast to the framework of \citet{GargKoenemann}, our formulation allows to stick to the natural edge-based formulations of the corresponding network flow problems and does not require an explicit reformulation of the problem as a packing problem.
		\item Moreover, the presented framework is the first application of the procedure of \citet{GargKoenemann} that natively supports the use of arbitrary linear objective functions, which allows the application to minimum cost flow problems. To the best of our knowledge, all prior applications instead transformed the objective functions into budget-constraints and searched for the optimal budget, which requires the restriction to positive costs and which results in weakly polynomial running times (cf. \citep{GargKoenemann,FleischerMulticommodity,WayneApproximation}).
	\end{itemize}

	The paper is structured as follows: In Section~\ref{sec:Preliminaries}, we briefly introduce the concepts that are inherent to our framework such as the procedure of \citet{GargKoenemann} and Megiddo's parametric search technique \citep{MegiddoCombinatorialOptimization,MegiddoParallel}. In Section~\ref{sec:Cones}, we reformulate the given problem~\eqref{eqn:OriginalProblem} as a packing problem and identify a subproblem that needs to be solved in each iteration of the algorithm. Moreover, we introduce three oracle types that will be investigated in the rest of the paper: a \emph{minimizing oracle} returning a cost-minimal vector in the ground set, a \emph{sign oracle} only returning a vector with the same sign as a cost minimal vector, and a \emph{separation oracle} either returning a vector with negative costs or stating that there is no such vector. Based on these considerations, we describe the general procedure in Section~\ref{sec:GeneralAlgorithm} and show that we can approximate problem~\eqref{eqn:OriginalProblem} efficiently if we are able to find a sufficiently good initial lower bound on the most violated dual constraint. In Section~\ref{sec:LowerBound}, we provide both weakly polynomial-time and strongly polynomial-time approaches to find such a lower bound. Finally, we apply our framework to a large class of network flow and packing problems in Section~\ref{sec:Applications}, including the maximum/minimum cost flow problem, generalized minimum cost flow problem, and the maximum/minimum cost flow problem in processing networks as well as budget-constrained versions of these. Moreover, we apply our framework to the maximum (weighted) spanning tree packing problem and the maximum (weighted) matroid base packing problem as examples of ``pure'' combinatorial problems. In Table~\ref{tab:Introduction:Results}, we give an overview of the results that will be derived in Section~\ref{sec:Applications}.

	\begin{table}[ht!]
		\def\tabularxcolumn#1{m{#1}}
		\def\arraystretch{1.5}
		\small
		\begin{tabularx}{\textwidth}{|X|X|X|}\hline
			\textbf{Problem} & \textbf{Previous best FPTAS} & \textbf{Our FPTAS} \\\hline
			Budget-Constrained Maximum Flow Problem & --- & $\OO\left(\frac{1}{\varepsilon^2} \cdot m \log m \cdot (m + n\log n) \right)$ \\\hline
			Budget-Constrained Minimum Cost Flow Problem & $\OOT\left(\frac{1}{\varepsilon^2} \cdot (n m^2 + n^3 m)) \right)$ \newline \citep{BudgetConstrainedComplexityApproximability} & $\OOT\left( \frac{1}{\varepsilon^2} \cdot n m^2 \right)$ \\\hline
			Budget-Constrained Minimum Cost Generalized Flow Problem & $\OOT\left( \frac{1}{\varepsilon^2} \cdot nm^2 \cdot (\log \frac{1}{\varepsilon} + \log\log M) \right)$ \newline (positive costs) \newline \citep{WayneApproximation} & $\OOT\left( \frac{1}{\varepsilon^2} \cdot n^2 m^2 \right)$ \newline (arbitrary costs) \\\hline
			Maximum Flows in Generalized Processing Networks & $\OO\left( \frac{1}{\varepsilon^2} \cdot m^2 \log m \right)$ \newline \citep{MaxFlowsInGeneralizedProcessingNetworks} & $\OO\left( \frac{1}{\varepsilon^2} \cdot m^2 \log m \right)$ \newline (simpler proof) \\\hline
			Minimum Cost Flows in Generalized Processing Networks & --- & $\OO\left( \frac{1}{\varepsilon^2} \cdot m^2 \log m \right)$ \\\hline
			Maximum Concurrent Flow Problem & $\OOT\left(\frac{1}{\varepsilon^2} \cdot (m^2 + kn) \right)$ \newline \citep{KarakostasConcurrentFlowProblem} & $\OOT\left(\frac{1}{\varepsilon^2} \cdot m^2 \cdot \min\{k,n\}\right)$ \newline (simpler proof) \\\hline
			Maximum Weighted Multicommodity Flow Problem & $\OOT\left( \frac{1}{\varepsilon^2} \cdot m^2 \cdot \min\{\log C, k\} \right)$ \newline \citep{FleischerMulticommodity} & $\OOT\left( \frac{1}{\varepsilon^2} \cdot m^2 \right)$ \\\hline
			Maximum Spanning Tree Packing Problem & $\OO\left(n^3m \log (n^2 / m) \right)$ \newline (exact algorithm) \newline \citep{GabowPackingTrees} & $\OO\left( \frac{1}{\varepsilon^2} \cdot m^2 \log m \cdot \alpha(m,n) \right)$ \\\hline
			Maximum Weighted Spanning Tree Packing Problem & --- & $\OO\left( \frac{1}{\varepsilon^2} \cdot m^2 \log m \cdot \alpha(m,n) \right)$ \\\hline
			Maximum Matroid Base Packing Problem & $\OO\left(m^7 \cdot F(m) \right)$ \newline (exact algorithm) \newline \citep[p. 734]{Schrijver2} & $\OO\left( \frac{1}{\varepsilon^2} \cdot m^2 \log m \cdot (F(m) + \log m) \right)$ \\\hline
			Maximum Weighted Matroid Base Packing Problem & --- & $\OOT\left( m \cdot F(m) \cdot \left(\frac{1}{\varepsilon^2} \cdot m + \log\log M \right) \right)$ \\\hline
		\end{tabularx}
		\caption{A summary of our results. Here, $m$ and $n$ denote the number of edges and nodes, respectively. The variable~$k$ denotes the number of commodities and $C$ denotes the largest ratio of any two weights of commodities. $M$ is the largest integer given in the input. The notation $\OOT(\cdot)$ ignores poly-logarithmic factors in $m$, so $\OOT(f(n,m)) = \OO(f(n,m) \cdot \log^q m)$ for some constant~$q$. $\alpha(m,n)$ denotes the inverse Ackermann function. $F(m)$ is the running time of a independence testing oracle for a given matroid with $m$ elements in its ground set.}
		\label{tab:Introduction:Results}
	\end{table}

	\section{Preliminaries}
	\label{sec:Preliminaries}

	\subsection{Approximation Algorithms}
	\label{sec:Pre:Approx}

	An algorithm~$A$ is called a (polynomial-time) \emph{approximation algorithm with performance guarantee $\alpha \in [1,\infty)$} or simply an \emph{$\alpha$--approximation} for a maximization problem~$\Pi$ with objective function~$c$ if, for each instance~$I$ of $\Pi$ with optimum solution~$x^*$, it computes a feasible solution~$x$ with objective value $c(x) \geq \frac{1}{\alpha} c(x^*)$ in polynomial time. An algorithm~$A$ that receives as input an instance~$I \in \Pi$ and a real number~$\varepsilon \in (0,1)$ is called a \emph{polynomial-time approximation scheme (PTAS)} if, on input~$(I,\varepsilon)$, it computes a feasible solution~$x$ with objective value $c(x) \geq (1 - \varepsilon) \cdot c(x^*)$ with a running-time that is polynomial in the encoding size~$|I|$ of $I$. If this running-time is additionally polynomial in $\frac{1}{\varepsilon}$, the algorithm is called a \emph{fully polynomial-time approximation scheme (FPTAS)}.

	\subsection{Garg and Koenemann's Fractional Packing Framework}
	\label{sec:Pre:GargKoenemann}

	Consider a fractional packing problem of the form~$\max\{c^T x : Ax \leq b, x \geq 0\}$ with non-negative entries in the matrix~$A \in \mathbb{N}^{m \times n}_{\geq 0}$. For example, we can model the maximum flow problem in such a way by interpreting the variables as flows on $s$-$t$-paths and requiring that the sum of flows on all paths that share some specific edge~$e$ is bounded by the edge's capacity. Hence, the vector~$b$ corresponds to the capacities of the edges and the vector~$c$ equals the all-one vector. Note that we need to stick to the path-based formulation of the problem since we are not allowed to introduce flow conservation constraints as they require negative coefficients. The dual formulation of the general primal problem is given as $\min\{b^T y : A^T y \geq c, y \geq 0\}$. In the example, the dual problem is to find small edge-lengths such that each path has length at least one. Although both the primal and the dual formulation of the problem are of exponential size in general, the fractional packing framework of \citet{GargKoenemann} allows us to obtain approximate solutions for the primal problem by using these formulations only implicitly, which will be shown in the following.

	Suppose that we want to find an $(1-\varepsilon)$-approximate solution for the primal problem with $\varepsilon \in (0,1)$ and let $\varepsilon' \colonequals \frac{\varepsilon}{2}$. The procedure described in \citep{GargKoenemann} starts with the feasible primal solution~$x = 0$ and the infeasible dual solution given by $y_i \colonequals \frac{\delta}{b_i} > 0$ for each $i \in \{1,\ldots,m\}$, where $\delta \colonequals \frac{(1+\varepsilon')}{\left((1+\varepsilon')m \right)^{\frac{1}{\varepsilon'}}}$. In each step of the algorithm, the \emph{most violated dual constraint} is determined based on the current dual solution~$y$, i.e., we determine a row~$j$ in the dual formulation that minimizes $\frac{(A^T)_{j\cdot} y}{c_j}$ among all rows with negative right-hand side value. For example, although there are exponentially many constraints in the dual formulation of the maximum flow problem, we can find the most violated constraint in $\OO(m + n\log n)$~time by computing a shortest $s$-$t$-path with edge lengths given by the dual vector~$y$. We then increase~$x_j$ by the (in terms of the primal problem) maximum allowed value~$\nu \colonequals \min_{i \in \{1,\ldots,m\}: A_{ij} > 0} \frac{b_i}{A_{ij}}$ (i.e., in the example, we send $\nu$ units of flow on the shortest path \emph{without} considering flow that has been sent in previous iterations), which will most likely make the primal solution infeasible. At the same time, each variable~$y_i$ will be multiplied by a factor of $\left(1 + \varepsilon' \cdot \frac{\nu}{\frac{b_i}{A_{ij}}} \right)$. Intuitively, for the maximum flow problem, the ``congested'' edges will get ``longer'' over time and will, thus, be used less likely in future iterations, which somehow balances the flow among all paths.

	The algorithm stops as soon as the dual solution fulfills $\sum_{i \in \{1,\ldots,m\}} b_i \cdot y_i \geq 1$. As noted above, the primal solution will most likely be infeasible since, in each iteration, the primal variables are increased regardless of the previous values. However, \citet{GargKoenemann} show that we obtain a feasible primal solution by scaling down the solution~$x$ by $\log_{1+\varepsilon'} \frac{1+\varepsilon'}{\delta}$ and that this solution is within a factor~$(1 - 2\varepsilon')=(1-\varepsilon)$ of the optimal solution. Moreover, they prove that the described procedure terminates within $\frac{1}{\varepsilon'} \cdot m \cdot (1 + \log_{1+\varepsilon'} m) = \OO\left(\frac{1}{\varepsilon^2} \cdot m \log m \right)$ iterations. We refer to \citep{GargKoenemann,FleischerMulticommodity,WayneApproximation} for further details on the procedure.

	In a follow-up publication, \citet{FleischerMulticommodity} showed that it suffices to determine an \emph{approximately most violated dual constraint} in each iteration of the problem: The claimed time bound and approximation guarantee continue to hold even if we choose a dual constraint $(A^T)_{j\cdot} y \geq c_j$ only fulfilling $\frac{(A^T)_{j\cdot} y}{c_j} \leq (1 + \varepsilon) \cdot \min_{l \in \{1,\ldots,n\}: c_l > 0} \frac{(A^T)_{l\cdot} y}{c_l}$. We will make use of this observation in Section~\ref{sec:GeneralAlgorithm}.

	\subsection{Megiddo's Parametric Search Technique}
	\label{sec:Pre:Megiddo}

	In Section~\ref{sec:Applications}, we make use of Megiddo's parametric search technique (cf. \citep{MegiddoCombinatorialOptimization}). Since we will go into the very heart of the method, we will briefly describe its idea in the following. We refer the reader to \citep{MegiddoCombinatorialOptimization} and \citep{MegiddoParallel} for further details on the method.

	Assume that we want to solve an optimization problem~$\Pi$ for which we already know an (exact) algorithm~$A$ that solves the problem, but in which some of the input values are now \emph{linear parametric values} that depend linearly on some real parameter~$\lambda$. Moreover, suppose that an algorithm~$C$ is known (in the following called \emph{callback}) that is able to decide if some candidate value for $\lambda$ is smaller, larger, or equal to the value~$\lambda^*$ that leads to an optimum solution to the underlying problem~$\Pi$. The idea of the parametric search technique is to extend the input of $A$ from constants to affine functions depending on $\lambda$ and to simulate the execution of algorithm~$A$ step by step. Note that each variable remains its affine structure as long as the algorithm only performs additions, subtractions, multiplications with constants, and comparisons. We call such an algorithm \emph{strongly combinatorial} in the following. Throughout the execution, an interval~$I$ is maintained that is known to contain the optimal value~$\lambda^*$. As soon as the simulation reaches the comparisons of two linear parametric values, since both values depend linearly on $\lambda$, it either holds that one of the variables is always larger than or equal to the other one in $I$ (in which case the result of the comparison is independent from $\lambda$) or that there is a unique intersection point~$\lambda'$. For this intersection point, we evaluate the callback~$C$ in order to determine if $\lambda' < \lambda^*$, $\lambda' > \lambda^*$, or $\lambda' = \lambda^*$ and, thus, resolve the comparison, update the interval~$I$, and continue the execution. Hence, as soon as the simulation of $A$ finishes, we have obtained an optimum solution to $\Pi$. The overall running-time is given by the running-time of $A$ times the running-time of $C$ and can be further improved using parallelization techniques described in \citep{MegiddoParallel}. We refer to \citep{MegiddoCombinatorialOptimization,MegiddoParallel} for details on the parametric search technique. Further applications and extensions of parametric search techniques can moreover be found in \citep{CohenFixedDimension,ToledoFixedDimension,ToledoApproximateParametricSearching}.

	\section{Packing over Cones}
	\label{sec:Cones}

	In this section, we transform the problem~\eqref{eqn:OriginalProblem} to a general (possibly exponential-sized) fractional packing problem in a first step and then reduce this problem to a more simple subproblem by incorporating the fractional packing framework of \citet{GargKoenemann}. Moreover, we introduce three types of oracles that enable us access to the cone~$C$ and that will be used throughout the rest of the paper.

	Let $S \colonequals \{ x^{(1)}, \ldots, x^{(k)} \}$ denote a finite set of $k$~non-negative $n$-dimensional vectors~$x^{(l)} \in \mathbb{R}^n_{\geq 0}$ with $x^{(l)} \neq 0$. The cone that is spanned by these vectors is given by
	\begin{align}\label{eqn:ConeDefinition}
		C \colonequals \left\{ x \in \mathbb{R}^n : x = \sum_{l=1}^k \alpha_l \cdot x^{(l)} \text{ with } \alpha_l \geq 0 \text{ for all } l \in \{1,\ldots,k\} \right\}.
	\end{align}
	The main result of this paper is that we are able to compute $(1-\varepsilon)$-approximate solutions for the problem to maximize a linear function over the cone~$C$ subject to a set of packing constraints under specific assumptions that will be investigated in the following. As we will see in Section~\ref{sec:Applications}, this extended framework is especially useful in the case of network flow problems for which some kind of flow decomposition theorem is known.

	Note that we do \emph{neither} require the set~$S$ to be of polynomial size \emph{nor} assume the set~$S$ or the cone~$C$ to be given explicitly. Instead, as it is common when dealing with implicitly given polyhedra, we only assume the cone to be \emph{well-described}, which implies that it has an encoding length of at least $n+1$ (cf. \citep{GroetschelLovaszSchrijver} for further details). Moreover, we make decisions over $S$ and $C$ via a given \emph{oracle}~$\mathcal{A}$ (to be specified later) that yields information about $S$ based on a given cost vector~$d$. 
	We assume that the running-time~$T_\mathcal{A}$ of each oracle~$\mathcal{A}$ fulfills~$T_\mathcal{A} \in \Omega(n)$ since it would not be able to investigate each component of $d$ or return a vector~$x \in C$ otherwise.


	In the following, let $A \in \mathbb{N}_{\geq 0}^{m \times n}$ denote a constraint matrix with non-negative entries, $b \in \mathbb{N}_{> 0}^m$ a positive right-hand side vector, and $c \in \mathbb{Z}^n$ a cost vector with arbitrary signs. Without loss of generality, we assume that at least one entry in each row and each column of $A$ is positive. Moreover, we define~$N$ to be the number of non-zero entries contained in the matrix~$A$.

	As described above, the problem we want to approximate is given as follows:%
	\begin{subequations}%
	\begin{align}%
		\max\ & c^T x \tag{\ref{eqn:OriginalProblem}a}\\
		\text{s.t.}\ & Ax \leq b, \tag{\ref{eqn:OriginalProblem}b}\\
		& x \in C. \tag{\ref{eqn:OriginalProblem}c}
	\end{align}%
	\end{subequations}%


	Using the definition of the cone~$C$ based on equation~\eqref{eqn:ConeDefinition}, we obtain the following equivalent formulation of the problem~\eqref{eqn:OriginalProblem}:
	\begin{align*}
		\max\ & c^T \sum_{l=1}^k \alpha_l \cdot x^{(l)} \\
		\text{s.t.}\ & A \left( \sum_{l=1}^k \alpha_l \cdot x^{(l)} \right) \leq b, \\
		& \alpha_l \geq 0 && \text{ for all } l \in \{ 1, \ldots, k \}.
	\end{align*}
	In particular, we replaced the original variables~$x$ by the weight vector~$\alpha$ and, in doing so, incorporated the constraints of the cone. As noted above, this formulation might be of exponential size. However, in the following, we will never need to state it explicitly but will derive results based on its implicit structure. By rearranging the objective function and the packing constraints, we obtain the following equivalent formulation of the problem:
	\begin{subequations}\label{eqn:Primal}
	\begin{align}
		\max\ & \sum_{l=1}^k \alpha_l \cdot \left( c^T x^{(l)} \right) \\
		\text{s.t.}\ & \sum_{l=1}^k \alpha_l \cdot \left( A_{i\cdot} x^{(l)} \right) \leq b_i && \text{ for all } i \in \{1,\ldots,m\}, \\
		& \alpha_l \geq 0 && \text{ for all } l \in \{ 1, \ldots, k \}.
	\end{align}
	\end{subequations}
	Clearly, we can neglect vectors~$x^{(l)}$ for which $c^T x^{(l)} \leq 0$ since, without loss of generality, it holds that $\alpha_l = 0$ for each such $l$ in an optimal solution. Hence, for the moment, we restrict\footnotemark\xspace our considerations on vectors~$x^{(l)}$ with $c^T x^{(l)} > 0$ such that the primal problem~\eqref{eqn:Primal} becomes in fact a fractional packing problem (again, possibly of exponential size). The dual formulation of this problem is given as follows:
	\footnotetext{We will see how we can ``filter out'' vectors~$x^{(l)}$ with negative costs in the following sections.}
	\begin{subequations}\label{eqn:Dual}
	\begin{align}
		\min\ & \sum_{i=1}^m y_i \cdot b_i \\
		\text{s.t.}\ & \sum_{i=1}^m y_i \cdot \left( A_{i\cdot} x^{(l)} \right) \geq c^T x^{(l)} && \text{ for all } l \in \{1,\ldots,k\} \text{ with } c^T x^{(l)} > 0, \label{eqn:DualConstraint}\\
		& y_i \geq 0 && \text{ for all } i \in \{ 1, \ldots, m \}.
	\end{align}
	\end{subequations}
	As it was shown in Section~\ref{sec:Pre:GargKoenemann}, we can apply the fractional packing framework of \cite{GargKoenemann} to the original problem~\eqref{eqn:OriginalProblem} \emph{provided} we are able to determine the most violated dual constraint in equation~\eqref{eqn:DualConstraint} efficiently. Hence, given a dual solution~$y > 0$, we need to be able to solve the following subproblem in polynomial time:
	\begin{align*}
		&\min_{\genfrac{}{}{0pt}{}{l \in \{1,\ldots,k\}}{c^T x^{(l)} > 0}} \dfrac{\sum_{i=1}^m y_i \cdot \left( A_{i\cdot} x^{(l)} \right)}{c^T x^{(l)}} %
		= \min_{\genfrac{}{}{0pt}{}{l \in \{1,\ldots,k\}}{c^T x^{(l)} > 0}} \dfrac{\sum_{i=1}^m y_i \cdot \sum_{j=1}^n A_{ij} \cdot x_j^{(l)}}{c^T x^{(l)}} \\
		=&\min_{\genfrac{}{}{0pt}{}{l \in \{1,\ldots,k\}}{c^T x^{(l)} > 0}} \dfrac{\sum_{j=1}^n x_j^{(l)} \cdot \sum_{i=1}^m y_i \cdot A_{ij}}{c^T x^{(l)}}.
	\end{align*}
	With $a_j \colonequals \sum_{i=1}^m y_i \cdot A_{ij}$ for $j \in \{1,\ldots,n\}$ and $a = (a_1,\ldots,a_n)^T$, this subproblem reduces to
	\begin{align}
		\min_{\genfrac{}{}{0pt}{}{l \in \{1,\ldots,k\}}{c^T x^{(l)} > 0}} \dfrac{a^T x^{(l)}}{c^T x^{(l)}}. \label{eqn:MostViolatedSubproblem}
	\end{align}
	Note that the vector~$a$ depends on $y$ and, thus, changes throughout the course of the procedure of \citet{GargKoenemann}. However, it always holds that $a_j > 0$ for each $j \in \{1,\ldots,n\}$ since $y_i > 0$ for each $i \in \{1,\ldots,m\}$ throughout the procedure and since the matrix~$A$ has at least one positive and no negative entry in each row as assumed above. Since $x^{(l)} \neq 0$ and $x^{(l)} \in \mathbb{R}^n_{\geq 0}$ for each $l \in \{1,\ldots,k\}$, this also yields that $a^T x^{(l)} > 0$, so the minimum in equation~\eqref{eqn:MostViolatedSubproblem} is always strictly positive.

	\begin{observation}\label{obs:EverythingPositive}
		It always holds that $a_j > 0$ for each $j \in \{1,\ldots,n\}$. Moreover, $a^T x^{(l)} > 0$ for each $x^{(l)} \in S$.
	\end{observation}

	Clearly, if the vectors in $S$ are given explicitly, we immediately obtain an FPTAS for the original problem~\eqref{eqn:OriginalProblem} using the arguments given in Section~\ref{sec:Pre:GargKoenemann}. In the following, we discuss the more elaborate case that we can access the set~$S$ and the cone~$C$ only via given oracles. Throughout this paper, we investigate three kinds of oracles with decreasing strength. The most powerful oracle to be considered can be defined as follows:

	\begin{definition}[Minimizing Oracle]\label{def:MinimizingOracle}
		For a given vector~$d \in \mathbb{R}^n$, a \emph{minimizing oracle} for the set $S$ returns a vector~$x^{(l^*)} \in S$ that minimizes $d^T x^{(l)}$ among all vectors~$x^{(l)} \in S$.
	\end{definition}

	Clearly, the notion of minimizing oracles requires very powerful algorithms. For example, if $S$ is the set of unit-flows on simple cycles in a given graph~$G$, a minimizing oracle would need to be able to determine a most negative simple cycle, which is $\NP$-complete in general (see Section~\ref{sec:Applications:BCMCFP}). In many cases, it suffices to consider a much weaker type of oracle given as follows:

	\begin{definition}[Sign Oracle]\label{def:SignOracle}
		For a given vector~$d \in \mathbb{R}^n$, a \emph{sign oracle} for the set $S$ returns a vector~$x^{(l)} \in S$ with\footnotemark\xspace $\sgn d^T x^{(l)} = \sgn d^T x^{(l^*)}$, where $x^{(l^*)}$ minimizes $d^T x^{(i)}$ among all vectors in $S$.
	\end{definition}
	\footnotetext{The \emph{sign function}~$\sgn \colon \mathbb{R} \mapsto \{-1,0,1\}$ returns $-1$, $0$, or $1$ depending on whether the argument is negative, zero, or positive, respectively.}

	Rather than determining a vector in $S$ with minimum cost, a sign oracle only returns \emph{any} vector whose cost have the same sign as a minimum-cost vector, which may be much easier to find. In the example above, we can easily find a cycle with the same costs as a most negative cycle by computing a minimum mean cycle in $\OO(nm)$~time (cf. \citep{KarpMinimumMeanCycle} and Section~\ref{sec:Applications:BCMCFP}). An even simpler kind of oracle is given as follows:

	\begin{definition}[Separation Oracle]\label{def:SeparationOracle}
		For a given vector~$d \in \mathbb{R}^n$, a \emph{separation oracle} for the set~$S$ either states that $d^T x^{(i)} \geq 0$ for all vectors~$x^{(i)} \in S$ or returns a \emph{certificate}~$x^{(l)} \in S$ that fulfills $d^T x^{(l)} < 0$.
	\end{definition}

	Clearly, the notion of separation oracles yields the least powerful yet most natural definition of an oracle. The name ``separation oracle'' is based on the fact that such an oracle can be seen as a traditional separation oracle for the \emph{dual cone}~$C^* \colonequals \{ w \in \mathbb{R}^n: w^T x \geq 0 \text{ for all } x \in C \}$ of the cone~$C$ (cf. \citep{GroetschelLovaszSchrijver}).

	Note that each minimizing oracle also induces a sign oracle and that each sign oracle induces a separation oracle, so the considered oracles have in fact decreasing strength. In particular, each approximation algorithm that is based on a sign oracle (separation oracle) is also valid for the case of a minimizing oracle (sign oracle).

	For the special case of uniform costs~$c^T x^{(l)}$ for all vectors~$x^{(l)} \in S$, we get the first approximation result based on the procedure of \citet{GargKoenemann}:

	\begin{theorem}\label{thm:FPTASMinimizing}
		Suppose that $c^T x^{(l)} = \hat{c}$ for all $x^{(l)} \in S$ and some constant~$\hat{c} > 0$. Given a minimizing oracle $\mathcal{A}$ for $S$ running in $T_\mathcal{A}$~time, a $(1 - \varepsilon)$-approximate solution for the problem~\eqref{eqn:OriginalProblem} can be computed in $\OO\left(\frac{1}{\varepsilon^2} \cdot m\log m \cdot  (N + T_\mathcal{A}) \right)$~time.
	\end{theorem}

	\begin{proof}
		Since $c^T x^{(l)} = \hat{c}$ for each $x^{(l)} \in S$, the subproblem given in equation~\eqref{eqn:MostViolatedSubproblem} reduces to the problem of finding a vector~$x^{(l)}$ with minimum cost~$\left(\frac{1}{\hat{c}} \cdot a\right)^T x^{(l)}$ among all vectors in $S$. Using the minimizing oracle, we can compute a minimizer for \eqref{eqn:MostViolatedSubproblem} in $\OO(T_\mathcal{A})$~time based on the cost vector~$d \colonequals \frac{1}{\hat{c}} \cdot a$. Note that this cost vector can be built in $\OO(N)$~time as each entry~$a_j$ of $a$ is defined to be $\sum_{i=1}^m y_i \cdot A_{ij}$, where each $y_i$ stems from the framework of \citet{GargKoenemann}. Consequently, we need look at each of the $N$ entries of $A$ once in order to build $d$. Hence, we are able to determine a most violated dual constraint of \eqref{eqn:Dual} in $\OO(N + T_\mathcal{A})$~time, so the claim follows immediately from the arguments outlined in Section~\ref{sec:Pre:GargKoenemann}.
	\end{proof}
	Note that the vector~$d$ that is constructed in the above procedure is always positive in each component according to Observation~\ref{obs:EverythingPositive}. As a consequence, if for example we use the vector~$d$ to denote the length of edges in a graph, we can use \citets{Dijkstra} algorithm to compute a shortest path. This will be used in Section~\ref{sec:Applications}.

	In the following sections, we will focus on the more general cases in which the costs are no longer uniform and in which we may only have access to the cone via weaker types of oracles.

	\section{General Algorithm}
	\label{sec:GeneralAlgorithm}

	Throughout this section, we assume that there is a separation oracle~$\mathcal{A}$ for $S$ running in $T_\mathcal{A}$~time. Hence, the presented results are valid for the case of minimizing oracles and sign oracles as well. In the subsequent section, we will see where the different strengths of the oracles come into play.

	The procedure of the upcoming algorithm is based on an idea introduced by \citet{FleischerMulticommodity}, which was originally developed for the maximum multicommodity flow problem: For $\lambda^*$ to denote the optimal value of the most violated dual constraint in equation~\eqref{eqn:MostViolatedSubproblem}, we let $\u{\lambda}$ denote a positive lower bound for $\lambda^*$. We will show in Section~\ref{sec:LowerBound} how we can find a good initial value for this lower bound efficiently. In each iteration of the procedure of \citet{GargKoenemann} as described in Section~\ref{sec:Pre:GargKoenemann}, we need to determine an approximately most violated dual constraint corresponding to some vector~$x^{(j)} \in S$ fulfilling
	\begin{align*}
		\frac{a^T x^{(j)}}{c^T x^{(j)}} \leq (1 + \varepsilon) \cdot \lambda^* = (1 + \varepsilon) \cdot \min_{\genfrac{}{}{0pt}{}{l \in \{1,\ldots,k\}}{c^T x^{(l)} > 0}} \dfrac{a^T x^{(l)}}{c^T x^{(l)}}.
	\end{align*}
	For $\lambda \in \mathbb{R}$, let $d(\lambda) \colonequals a - \lambda c$ and $D(\lambda) \colonequals \min\{ d(\lambda)^T x^{(l)} : l \in \{1,\ldots,k\} \text{ with } c^T x^{(l)} > 0 \}$. Similar to the minimum ratio cycle problem \citep{LawlerCombinatorialOptimization,MegiddoCombinatorialOptimization,MegiddoParallel}, we get the following characterization of the relation between the sign of $D(\lambda)$ and the sign of $\lambda^* - \lambda$:

	\begin{lemma}\label{lem:DecisionParametric}
		For some given value of $\lambda \in \mathbb{R}$, it holds that $\sgn(D(\lambda)) = \sgn(\lambda^* - \lambda)$.
	\end{lemma}

	\begin{proof}
		Let $L \colonequals \{ l \in \{1,\ldots,k\} : c^T x^{(l)} > 0 \}$. First, consider the case that $D(\lambda) > 0$. The claim follows by simple arguments:
		\begin{align*}
			D(\lambda) > 0 &\Longleftrightarrow d(\lambda)^T x^{(l)} > 0 &&\text{ for all } l \in L \\
			&\Longleftrightarrow (a - \lambda c)^T x^{(l)} > 0 &&\text{ for all } l \in L \\
			&\Longleftrightarrow \frac{a^T x^{(l)}}{c^T x^{(l)}} > \lambda &&\text{ for all } l \in L \\
			&\Longleftrightarrow \lambda^* > \lambda.
		\end{align*}
		Conversely, if $D(\lambda) < 0$, we get the following equivalences by similar arguments:
		\begin{align*}
			D(\lambda) < 0 &\Longleftrightarrow d(\lambda)^T x^{(l)} < 0 &&\text{ for some } l \in L \\
			&\Longleftrightarrow (a - \lambda c)^T x^{(l)} < 0 &&\text{ for some } l \in L \\
			&\Longleftrightarrow \frac{a^T x^{(l)}}{c^T x^{(l)}} < \lambda &&\text{ for some } l \in L \\
			&\Longleftrightarrow \lambda^* < \lambda.
		\end{align*}
		Finally, in the remaining case $D(\lambda) = 0$, it follows by continuity that $\lambda^* = \lambda$, which shows the claim.
	\end{proof}

	Lemma~\ref{lem:DecisionParametric} implies that $\lambda^*$ is the maximum value of $\lambda$ such that $D(\lambda) \geq 0$, i.e., such that $d(\lambda)^T x^{(l)} \geq 0$ for each $x^{(l)} \in S$ with $c^T x^{(l)} > 0$. In each iteration of our general procedure, we call the given separation oracle~$\mathcal{A}$ with the vector~$d((1+\varepsilon)\u{\lambda})$. We distinguish between the two possible outcomes of one such call:
	\begin{description}
		\item[\textbf{Case 1:}] The oracle returns some certificate~$x^{(l)} \in S$ with $d((1+\varepsilon)\u{\lambda})^T x^{(l)} < 0$. In this case, we get that
			\begin{align*}
										&& \left( a - (1 + \varepsilon) \cdot \u{\lambda} \cdot c \right)^T x^{(l)} &< 0 \\
				\Longleftrightarrow 	&& a^T x^{(l)} &< (1 + \varepsilon) \cdot \u{\lambda} \cdot c^T x^{(l)} \\
				\Longrightarrow 		&& \frac{a^T x^{(l)}}{c^T x^{(l)}} &< (1 + \varepsilon) \cdot \u{\lambda} \\
				\Longrightarrow 		&& \frac{a^T x^{(l)}}{c^T x^{(l)}} &< (1 + \varepsilon) \cdot \lambda^*.
			\end{align*}
			The third inequality follows from the fact that $a^T x^{(l)} > 0$ (cf. Observation~\ref{obs:EverythingPositive}) and that $\u{\lambda} > 0$ such that it also holds that $c^T x^{(l)} > 0$. The returned vector~$x^{(l)}$ yields an approximately most violated dual constraint. We use this dual constraint and continue the procedure of \citet{GargKoenemann}. Note that, during an iteration of the procedure, it holds that $c^T x^{(j)}$ remains constant for each vector~$x^{(j)} \in S$ (since it does not depend on the dual solution~$y$) and that $a^T x^{(l)}$ does not decrease (since both the entries in $x^{(l)}$ and the entries in $A$ are non-negative). Hence, $\u{\lambda}$ continues to be a lower bound for the (possibly increased) new value~$\lambda^*$ of \eqref{eqn:MostViolatedSubproblem}.
		\item[\textbf{Case 2:}] The oracle states that all vectors~$x^{(l)}$ fulfill $d((1+\varepsilon)\u{\lambda})^T x^{(l)} \geq 0$. It now holds that
			\begin{align*}
										&& \left( a - (1 + \varepsilon) \cdot \u{\lambda} \cdot c \right)^T x^{(l)} &\geq 0 && \forall x^{(l)} \in S \\
				\Longleftrightarrow 	&& a^T x^{(l)} &\geq (1 + \varepsilon) \cdot \u{\lambda} \cdot c^T x^{(l)}  && \forall x^{(l)} \in S \\
				\Longleftrightarrow 	&& \frac{a^T x^{(l)}}{c^T x^{(l)}} &\geq (1 + \varepsilon) \cdot \u{\lambda}  && \forall x^{(l)} \in S \text{ with } c^T x^{(l)} > 0 \\
				\Longleftrightarrow 	&& \lambda^* &\geq (1 + \varepsilon) \cdot \u{\lambda}.
			\end{align*}
			In this case, we can update the lower bound~$\u{\lambda}$ to $(1 + \varepsilon) \cdot \u{\lambda}$ and continue.
	\end{description}

	Hence, in each iteration of the algorithm, we either proceed in the procedure of \citet{GargKoenemann} or we increase the lower bound by a factor of $1 + \varepsilon$. Again, we want to stress that $\u{\lambda}$ continues to be a lower bound for $\lambda^*$ after an iteration of the above algorithm. The following theorem shows that this yields an efficient approximation algorithm for the problem~$\eqref{eqn:OriginalProblem}$ provided we are given a sufficiently good initial estimate for $\lambda^*$:

	\begin{theorem}\label{thm:FPTASforGivenInitialValue}
		Given a separation oracle~$\mathcal{A}$ for $S$ running in $T_\mathcal{A}$~time and given an initial lower bound~$\u{\lambda}$ for the initial value of $\lambda^*$ fulfilling $\u{\lambda} \leq \lambda^* \leq m^{\frac{1}{\varepsilon} m} \cdot \u{\lambda}$, a $(1 - \varepsilon)$-approximate solution for problem~\eqref{eqn:OriginalProblem} can be determined in $\OO\left(\frac{1}{\varepsilon^2} \cdot m \log m \cdot (N + T_\mathcal{A}) \right)$~time.
	\end{theorem}

	\begin{proof}
		The correctness of the procedure follows from the arguments outlined in Section~\ref{sec:Pre:GargKoenemann}, the preceding discussion, and the fact that the initial value of $\u{\lambda}$ is a valid lower bound for $\lambda^*$.

		In each step of the algorithm, we evaluate the given separation oracle and -- based on its result -- either perform one iteration of the procedure of \citet{GargKoenemann} and \citet{FleischerMulticommodity} or update the lower bound~$\u{\lambda}$. As noted in Section~\ref{sec:Pre:GargKoenemann}, the former case occurs up to $\OO(\frac{1}{\varepsilon^2} \cdot m \log m)$~times. In order to determine the number of updates to $\u{\lambda}$, let $(\lambda^*)^{(k)}$, $\u{\lambda}^{(k)}$, $y_i^{(k)}$ denote the values of the corresponding variables after the $k$-th iteration of the overall algorithm and let $\tau$ denote the number of iterations until the algorithm stops. Note that the procedure stops as soon as $\sum_{i=1}^m b_i \cdot y^{(k)}_i \geq 1$. Hence, after the $(\tau-1)$-th iteration, it holds that $y^{(\tau-1)}_i < \frac{1}{b_i}$ for each $i \in \{1,\ldots,m\}$. Since each variable~$y^{(\tau-1)}_i$ will be increased by a factor of at most $1 + \varepsilon$ in the final iteration, it holds that $y^{(\tau)}_i < (1 + \varepsilon) \cdot \frac{1}{b_i}$. Since the initial value of each variable~$y_i$ was set to $y^{(0)}_i \colonequals \frac{\delta}{b_i}$, every dual variable increases by a factor of at most $\frac{1 + \varepsilon}{\delta}$ during the execution of the algorithm, so $y^{(\tau)}_i \leq \frac{1 + \varepsilon}{\delta} \cdot y^{(0)}$. However, this also implies that $(\lambda^*)^{(\tau)} \leq \frac{1 + \varepsilon}{\delta} \cdot (\lambda^*)^{(0)}$. Since the lower bound~$\u{\lambda}$ remains to be a lower bound after every step of the algorithm as discussed above, it holds that
		\[
			\u{\lambda}^{(\tau)} \leq (\lambda^*)^{(\tau)} \leq \frac{1 + \varepsilon}{\delta} \cdot (\lambda^*)^{(0)} \leq \frac{1 + \varepsilon}{\delta} \cdot m^{\frac{1}{\varepsilon} m} \cdot \u{\lambda}^{(0)},
		\]
		where the third inequality follows from the requirement that $(\lambda^*)^{(0)} \leq m^{\frac{1}{\varepsilon} m} \cdot \u{\lambda}^{(0)}$. Since $\u{\lambda}$ is increased by a factor of $1 + \varepsilon$ in each update step, we get that the number of such steps is bounded by
		\begin{align*}
			\log_{1 + \varepsilon} \dfrac{\u{\lambda}^{(\tau)}}{\u{\lambda}^{(0)}}
			&\leq \log_{1 + \varepsilon} \left( \dfrac{1 + \varepsilon}{\delta} \cdot m^{\frac{1}{\varepsilon} m} \right)
			= \log_{1 + \varepsilon} \dfrac{1 + \varepsilon}{\delta} + \log_{1 + \varepsilon} m^{\frac{1}{\varepsilon} m} \\
			&= \log_{1 + \varepsilon} \dfrac{1 + \varepsilon}{\frac{1+\varepsilon}{((1+\varepsilon)m)^{\frac{1}{\varepsilon}}}} + \frac{1}{\varepsilon} \cdot m \log_{1 + \varepsilon} m \\
			&= \log_{1 + \varepsilon} ((1 + \varepsilon) m)^\frac{1}{\varepsilon}  + \frac{1}{\varepsilon} \cdot m \log_{1 + \varepsilon} m \\
			&= \frac{1}{\varepsilon} \cdot (1 + \log_{1 + \varepsilon} m)  + \frac{1}{\varepsilon} \cdot m \log_{1 + \varepsilon} m \\
			&\in \OO\left( \frac{1}{\varepsilon} \cdot m \log_{1+\varepsilon} m \right)
			= \OO\left( \frac{1}{\varepsilon} \cdot m \cdot \frac{\log m}{\log (1+\varepsilon)} \right)
			= \OO\left( \frac{1}{\varepsilon^2} \cdot m \log m \right)
		\end{align*}
		and, thus, matches the number of iterations of the procedure of \citet{GargKoenemann}. The claim then follows by the fact that, in each step of the algorithm, we need $\OO(N)$~time to compute the entries of the vector~$d((1+\varepsilon)\u{\lambda})$ and $T_\mathcal{A}$~time to evaluate the oracle.
	\end{proof}

	Note that the allowed deviation of the initial lower bound~$\u{\lambda}$ to $\lambda^*$ in Theorem~\ref{thm:FPTASforGivenInitialValue} is chosen in a way such that the number of update steps to $\u{\lambda}$ does not dominate the $\OO(\frac{1}{\varepsilon^2} \cdot m \log m)$~steps of the overall procedure.

	\section{Determining a lower bound}
	\label{sec:LowerBound}

	The proof of Theorem~\ref{thm:FPTASforGivenInitialValue} shows that the strongly polynomial number of oracle calls depends on the assumption that the initial value for the lower bound~$\u{\lambda}$ is not ``too far away'' from the real value~$\lambda^*$ of the most violated dual constraint. In this section, we present a weakly polynomial-time and a strongly polynomial-time approach to find such a sufficiently good initial value.

	\subsection{Weakly Polynomial-Time Approach for Separation Oracles}
	\label{sec:LowerBoundWeakly}

	We start by providing a general approach that is valid for all three types of oracles. The running time will depend (logarithmically) on the largest number given in the input, denoted by $M \colonequals \max\{ (\max_i b_i), (\max_j c_j), (\max_{i,j} A_{i,j}), n, m\} \in \mathbb{N}$.

	\begin{lemma}\label{lem:InitialBoundWeakly}
		Suppose that we are given a separation oracle~$\mathcal{A}$ running in $T_\mathcal{A}$~time. An initial lower bound~$\u{\lambda}$ for $\lambda^*$ fulfilling $\u{\lambda} \leq \lambda^* \leq m^{\frac{1}{\varepsilon} m} \cdot \u{\lambda}$ can be determined in weakly polynomial time $\OO((T_\mathcal{A} + N) \cdot (\log\log M - (\log \frac{1}{\varepsilon} + \log m + \log \log m)))$.
	\end{lemma}

	\begin{proof}
		Let $x^{(l)} \in S$ denote a vector with $\frac{a^T x^{(l)}}{c^T x^{(l)}} = \lambda^*$ that determines the minimum in equation~\eqref{eqn:MostViolatedSubproblem}. Using that $y_i \colonequals \frac{b_i}{\delta}$ for each $i \in \{1,\ldots,m\}$ at the beginning of the procedure, we get that
		\begin{align}
			\frac{a^T x^{(l)}}{c^T x^{(l)}}
			&= \frac{\sum_{j=1}^n a_j \cdot x^{(l)}_j}{\sum_{j=1}^n c_j \cdot x_j^{(l)}}
			= \frac{\sum_{j=1}^n \left( \sum_{i=1}^m y_i \cdot A_{ij} \right) \cdot x^{(l)}_j}{\sum_{j=1}^n c_j \cdot x_j^{(l)}}
			= \frac{\sum_{j=1}^n \sum_{i=1}^m \frac{\delta}{b_i} \cdot A_{ij} \cdot x^{(l)}_j}{\sum_{j=1}^n c_j \cdot x_j^{(l)}}. \label{eqn:SubproblemFormulaToBound}
		\end{align}
		Without loss of generality, we can assume the separation oracle~$\mathcal{A}$ to always return a vector~$x^{(l)} \in S$ with $\max_{j \in \{1,\ldots,n\}} x_j = 1$ (whenever it returns a vector at all): For each $x \in C$, it also holds that $\beta \cdot x \in C$ for some positive constant~$\beta$. Hence, if the oracle does not fulfill the required property, we can wrap it into a new oracle~$\mathcal{A'}$ which divides the vector returned by $\mathcal{A}$ by $\max_{j \in \{1,\ldots,n\}} x_i > 0$. Using this fact in equation~\eqref{eqn:SubproblemFormulaToBound}, we get the following lower and upper bound on $\lambda^*$:
		\begin{align*}
			\lambda^* &\geq \frac{\delta}{M} \cdot \frac{\sum_{j=1}^n \left(\sum_{i=1}^m A_{ij} \right) \cdot x^{(l)}_j}{\sum_{j=1}^n M \cdot x_j^{(l)}}
			\geq \frac{\delta}{M} \cdot \frac{\sum_{j=1}^n 1 \cdot x^{(l)}_j}{n \cdot M}
			\geq \frac{\delta}{n \cdot M^2} \geq \frac{\delta}{M^3} \equalscolon \lambda_1
		\intertext{and}
			\lambda^* &\leq \frac{\delta}{1} \cdot \frac{\sum_{j=1}^n \left(\sum_{i=1}^m A_{ij} \right) \cdot x^{(l)}_j}{\sum_{j=1}^n 1 \cdot x_j^{(l)}}
			\leq \delta \cdot \frac{\sum_{j=1}^n m \cdot M \cdot x^{(l)}_j}{1}
			\leq \delta \cdot nm \cdot M \leq \delta \cdot M^3 \equalscolon \lambda_2.
		\end{align*}
 
		According to Lemma~\ref{lem:DecisionParametric}, each feasible lower bound~$\u{\lambda}$ for $\lambda^*$ is characterized by the fact that $D(\u{\lambda}) \geq 0$, so an oracle call with the vector~$d(\u{\lambda})$ results in the answer that there are no vectors in $S$ with negative costs. Since $\u{\lambda}$ is required to fulfill $\u{\lambda} \leq \lambda^* \leq m^{\frac{1}{\varepsilon} m} \cdot \u{\lambda}$, we only need to consider values for $\u{\lambda}$ of the form~$\lambda_1 \cdot (m^{\frac{1}{\varepsilon} m})^k$ in $[\lambda_1,\lambda_2]$ for integral values of $k$. Moreover, since the oracle returns a vector if and only if $\u{\lambda} > \lambda^*$, we can perform a binary search on these values in order to find the best possible lower bound for $\lambda^*$. In total, we get the following number of iterations:
		\begin{align*}
			\OO\left(\log\log_{m^{\frac{1}{\varepsilon} m}} \frac{\lambda_2}{\lambda_1} \right)
			&= \OO\left(\log\log_{m^{\frac{1}{\varepsilon} m}} \dfrac{\delta \cdot M^3}{\frac{\delta}{M^3}} \right)
			= \OO\left(\log\log_{m^{\frac{1}{\varepsilon} m}} M \right) \\
			&= \OO\left(\log \frac{\log M}{\log m^{\frac{1}{\varepsilon} m}} \right)
			= \OO\left(\log \frac{\log M}{\frac{1}{\varepsilon} \cdot m \log m} \right) \\
			&= \OO\left(\log\log M - \left(\log \frac{1}{\varepsilon} + \log m + \log\log m \right) \right).
		\end{align*}
		In combination with the overhead of $N + T_\mathcal{A}$ to call the oracle (as in the proof of Theorem~\ref{thm:FPTASforGivenInitialValue}), we get the claimed time bound.
	\end{proof}

	Note that the time bound given in Lemma~\ref{lem:InitialBoundWeakly} is in fact only weakly polynomial for very large values of $M$: The binary search only has an effect on the overall running time if the encoding length~$\log M$ of $M$ fulfills $\log M \in \omega(\frac{1}{\varepsilon} \cdot m \log m)$, i.e., if $M$ is exponential in $\frac{1}{\varepsilon} \cdot m \log m$.

	Theorem~\ref{thm:FPTASforGivenInitialValue} in combination with Lemma~\ref{lem:InitialBoundWeakly} yields the following theorem:

	\begin{theorem}\label{thm:FPTASWeakly}
		Given a separation oracle~$\mathcal{A}$ for $S$ running in $T_\mathcal{A}$~time, a $(1-\varepsilon)$-approximate solution for problem~\eqref{eqn:OriginalProblem} can be determined in weakly polynomial time $\OO((T_\mathcal{A} + N) \cdot (\frac{1}{\varepsilon^2} \cdot m \log m + \log\log M - (\log \frac{1}{\varepsilon} + \log m + \log \log m)))$.\qed
	\end{theorem}

	In particular, if the oracle~$\mathcal{A}$ runs in polynomial time, we immediately obtain an FPTAS for problem~\eqref{eqn:OriginalProblem} according to Theorem~\ref{thm:FPTASWeakly}.

	\subsection{Strongly Polynomial-Time Approach for Sign Oracles}
	\label{sec:LowerBoundStronglySign}

	In the previous subsection, we introduced a method to determine an initial lower bound for $\lambda^*$ that is valid for each of the investigated types of oracles. However, although the general procedure that was described in Section~\ref{sec:GeneralAlgorithm} performs a strongly polynomial number of steps, the overall procedure would not yield a strongly polynomial FPTAS, in general, even if the oracle runs in strongly polynomial time due to the weakly polynomial overhead of the binary search. In this section, we present an alternative method for minimizing and sign oracles running in strongly polynomial time. In the subsequent subsection, we generalize the result to separation oracles.

	According to Lemma~\ref{lem:DecisionParametric}, we can decide about the direction of the deviation between some candidate value~$\lambda$ and the desired value~$\lambda^*$, if we are able to determine the sign of $D(\lambda)$. Clearly, this task is strongly related to the definition a sign oracle for $S$. However, the value $D(\lambda)$ is defined to be the minimum of $d(\lambda)^T x^{(l)}$ among all vectors~$x^{(l)}$ that \emph{additionally} fulfill $c^T x^{(l)} > 0$ whereas the sign oracle is not required to consider only such vectors according to Definition~\ref{def:SignOracle}. Nevertheless, as it will be shown in the following lemma, we can neglect this additional restriction when evaluating the sign oracle:

	\begin{lemma}\label{lem:SignOracleForDecision}
		For any positive value of $\lambda$, it holds that $\sgn(D(\lambda)) = \sgn(d(\lambda)^T x^{(l)})$ where $x^{(l)}$ is a vector returned by a sign oracle for $S$.
	\end{lemma}

	\begin{proof}
		First consider the case that $\sgn(d(\lambda)^T x^{(l)}) = -1$, i.e., that $d(\lambda)^T x^{(l)} < 0$. Using the definition of $d(\lambda)$, we get that $(a - \lambda c)^T x^{(l)} = a^T x^{(l)} - \lambda \cdot c^T x^{(l)} < 0$. Since both $a^T x^{(l)} > 0$ according to Observation~\ref{obs:EverythingPositive} and $\lambda > 0$, it must hold that $c^T x^{(l)} > 0$ as well. Thus, we can conclude that $D(\lambda) \leq d(\lambda)^T x^{(l)} < 0$.

		Now consider the case that $\sgn(d(\lambda)^T x^{(l)}) = 0$. According to Definition~\ref{def:SignOracle}, it holds that there are no vectors~$x^{(j)} \in S$ with $d(\lambda)^T x^{(j)} < 0$, so $D(\lambda) \geq 0$. As in the previous case, we get that $(a - \lambda c)^T x^{(l)} = a^T x^{(l)} - \lambda \cdot c^T x^{(l)} = 0$ if and only if $c^T x^{(l)} > 0$ since both $a^T x^{(l)} > 0$ and $\lambda > 0$. Hence, we also get that $D(\lambda) \leq d(\lambda)^T x^{(l)} = 0$, so $D(\lambda) = 0$.

		Finally, if $\sgn(d(\lambda)^T x^{(l)}) = 1$, there are no vectors $x^{(i)} \in S$ with $d(\lambda)^T x^{(i)} \leq 0$. Thus, it also holds that $D(\lambda) > 0$, which shows the claim.
	\end{proof}

	Lemma~\ref{lem:SignOracleForDecision} now allows us to determine a sufficiently good initial lower bound~$\u{\lambda}$. In fact, as it will be shown in the following lemma, we are even able to determine an \emph{exact} most violated dual constraint in \emph{each} iteration of the procedure:

	\begin{lemma}\label{lem:SignOracleMostViolatedMegiddo}
		Given a strongly combinatorial and strongly polynomial-time sign oracle~$\mathcal{A}$ for $S$ running in $T_\mathcal{A}$~time, a most violated dual constraint can be determined in $\OO\left(N + T^2_\mathcal{A} \right)$~time.
	\end{lemma}

	\begin{proof}
		Lemma~\ref{lem:DecisionParametric} and Lemma~\ref{lem:SignOracleForDecision} imply that $\lambda^*$ is the unique value for $\lambda$ for which the sign oracle returns a vector~$x^{(l)} \in S$ with $d(\lambda)^T x^{(l)} = 0$. In particular, the returned vector~$x^{(l)}$ is a minimizer for \eqref{eqn:MostViolatedSubproblem}. Hence, since the values~$a_i$ can be computed in $\OO(N)$~time, we are done if we are able to determine such a vector~$x^{(l)}$ in $\OO(T^2_\mathcal{A})$~time.

		Let $d(\lambda)$ be defined as above, where $\lambda$ is now treated as a \emph{symbolic} value that is known to be contained in an interval~$I$. Initially, we set~$I$ to $(0,+\infty)$ since the optimal value~$\lambda^*$ is known to be strictly positive (cf. equation~\eqref{eqn:MostViolatedSubproblem}). Note that the costs~$(d(\lambda))_i = a_i - \lambda \cdot c_i$ fulfill the linear parametric value property. We simulate the execution of the sign oracle~$\mathcal{A}$ at input~$d(\lambda)$ using \citets{MegiddoCombinatorialOptimization} parametric search technique as described in Section~\ref{sec:Pre:Megiddo}. The underlying idea is to ``direct'' the control flow during the execution of $\mathcal{A}$ in a way such that it eventually returns the desired vector minimizing \eqref{eqn:MostViolatedSubproblem}.

		Whenever we need to resolve a comparison between two linear parametric values that intersect at some point~$\lambda' \in I$, we call the sign oracle itself with the cost vector~$d \colonequals d(\lambda')$ in order to determine the sign of $D(\lambda')$. If $D(\lambda') = 0$, we found a minimizer for equation~\eqref{eqn:MostViolatedSubproblem} and are done. If $D(\lambda') < 0$ ($D(\lambda') > 0$), the candidate value~$\lambda'$ for $\lambda^*$ was too large (too small) according to Lemma~\ref{lem:DecisionParametric} and Lemma~\ref{lem:SignOracleForDecision} such that we can update the interval~$I$ to $I \cap (-\infty, \lambda')$ ($I \cap (\lambda', +\infty)$), resolve the comparison, and continue the simulation of the oracle algorithm. After $\OO(T_\mathcal{A})$ steps, the simulation terminates and returns a vector~$x^{(l)} \in S$ that fulfills~$d(\lambda^*)^T x^{(l)} = 0$ for the desired value~$\lambda^* \in I$. Hence, this vector yields a most violated constraint in \eqref{eqn:DualConstraint}. Since the described simulation runs in $\OO(T_\mathcal{A}^2)$~time, the claim follows.
	\end{proof}

	Note that we actually still obtain a polynomial running-time of the above algorithm even if we do not assume the sign oracle to run in \emph{strongly} polynomial time but only to run in \emph{weakly} polynomial time. However, the running-time of the resulting algorithm might exceed the stated time bound since the candidate values~$\lambda'$ that determine the input to the callback oracle are rational numbers whose representation might involve exponential-size numbers of the form $H^{T_\mathcal{A}}$ for some~$H$ with polynomial encoding length. Although the running-time of a weakly polynomial-time oracle algorithm depends only logarithmically on the size of these numbers, the running-time might still increase by a large (polynomial) factor.

	Lemma~\ref{lem:SignOracleMostViolatedMegiddo} can be incorporated into the procedure of \citet{GargKoenemann} to obtain an FPTAS for problem~\eqref{eqn:OriginalProblem} running in $\OO(\frac{1}{\varepsilon^2} \cdot m \log m \cdot (N + T^2_\mathcal{A}))$~time. However, it can also be used to find an initial lower bound~$\u{\lambda}$ for $\lambda^*$ (which, in fact, equals $\lambda^*$), which yields the following theorem in combination with Theorem~\ref{thm:FPTASforGivenInitialValue}:

	\begin{theorem}\label{thm:FPTASStronglySign}
		Given a strongly combinatorial and strongly polynomial-time sign oracle~$\mathcal{A}$ for $S$ running in $T_\mathcal{A}$~time, there is a strongly polynomial FPTAS for the problem~\eqref{eqn:OriginalProblem} running in $\OO\left(\frac{1}{\varepsilon^2} \cdot m\log m \cdot \left(N + T_\mathcal{A} \right) + T^2_\mathcal{A} \right)$ time. \qed
	\end{theorem}

	\subsection{Strongly polynomial-time approach for Separation Oracles}
	\label{sec:LowerBoundStronglySeparation}

	Although separation oracles are probably the most natural kind of oracle, they are also the weakest of the considered oracle types. The proof of Lemma~\ref{lem:SignOracleMostViolatedMegiddo} relies on the fact that we are able to decide if some candidate value~$\lambda$ is too small, too large, or equal to the desired value. In the case of a separation oracle, however, the case that $d^T x^{(i)} \geq 0$ for all vectors~$x^{(i)} \in S$ does no longer include the information whether there is a vector~$x^{(l)} \in S$ with $d^T x^{(l)} = 0$ (in which case we have found the desired vector in the parametric search as described above) or if $d^T x^{(i)} > 0$ for all $x^{(i)} \in S$. For example, if we come across a comparison of the form~$a_0 + \lambda \cdot a_1 \leq b_0 + \lambda \cdot b_1$ during the simulation where $a_1 > b_1$, we are actually interested in the information whether or not the optimal value~$\lambda^*$ fulfills $\lambda^* \leq \lambda' \colonequals \frac{b_0 - a_0}{a_1 - b_1}$. However, if we use the separation oracle with the cost vector~$d(\lambda')$, we only obtain the information whether $\lambda^* < \lambda'$ (in case that the oracle returns a certificate) or if $\lambda^* \geq \lambda'$. Hence, in the latter case, the outcome of the comparison is not yet resolved since we need the additional information whether or not $\lambda^* = \lambda'$, so we cannot continue the simulation without any further ado. Nevertheless, as it will be shown in the following lemma, we can gather this additional information by a more sophisticated approach:

	\begin{lemma}\label{lem:SeparationOracleMostViolatedMegiddo}
		Given a strongly combinatorial and strongly polynomial-time separation oracle~$\mathcal{A}$ for $S$ running in $T_\mathcal{A}$~time, a most violated dual constraint can still be determined in $\OO\left(N + T^2_\mathcal{A} \right)$ time.
	\end{lemma}

	\begin{proof}
		The claim directly follows from Lemma~\ref{lem:SignOracleMostViolatedMegiddo} if we can show that we can extend the given separation oracle into a sign oracle for $S$. As in the proof of Lemma~\ref{lem:SignOracleMostViolatedMegiddo}, we simulate the execution of the separation oracle using the parametric cost vector~$d(\lambda) \colonequals a - \lambda c$. Assume that the execution halts at a comparison that needs to be resolved, resulting in a candidate value~$\lambda'$ for the desired value~$\lambda^*$. We invoke the separation oracle with the cost vector~$d \colonequals d(\lambda')$. Clearly, if the oracle returns a certificate~$x^{(l)}$ with $d^T x^{(l)} < 0$, we can conclude that $D(\lambda') < 0$ such that the value~$\lambda'$ was too large according to Lemma~\ref{lem:DecisionParametric} and the result of the comparison is determined. Conversely, if the oracle states that $d^T x^{(i)} \geq 0$ for all $x^{(i)} \in S$, we can conclude that $D(\lambda') \geq 0$. However, we may not yet be able to resolve the comparison since its result may rely on the additional information whether $D(\lambda') = 0$ or $D(\lambda') > 0$ as shown above. Nevertheless, we can extract this information by one additional call to the oracle as it will be shown in the following.

		First suppose that $\lambda' = \lambda^*$. In this situation, it holds that $d(\lambda')^T x^{(i)} \geq 0$ for all $x^{(i)} \in S$ and there is at least one vector~$x^{(l)} \in S$ that fulfills $d(\lambda')^T x^{(l)} = 0$. Since all the functions~$f^{(i)}(\lambda) \colonequals d(\lambda)^T x^{(i)} = a^T x^{(i)} - \lambda \cdot c^T x^{(i)}$ are linear functions of $\lambda$ with negative slope (in case that $c^T x^{(i)} > 0$; otherwise, the function has no positive root at all), it holds that several functions~$f^{(l)}$ evaluate to zero at $\lambda'$ while every other function attains its root at a larger value for $\lambda$ (cf. Figure~\ref{fig:ParametricFunctionEqual}). Hence, for every larger value of $\lambda$, the separation oracle changes its outcome and returns a certificate. In particular, for a sufficiently small but positive value of $\delta$, the separation oracle called with the cost vector~$d(\lambda' + \delta)$ returns a vector~$x^{(l)} \in S$ with $d(\lambda' + \delta) x^{(l)} < 0$ that additionally fulfills $d(\lambda')^T x^{(l)} = 0$ (so $x^{(l)}$ yields a most violated constraint in the overall procedure). Clearly, the value of $\delta$ must be small enough to guarantee that we do not reach the root of another function~$f^{(i)}$ (i.e., smaller than the distance between the dashed and the dotted line in Figure~\ref{fig:ParametricFunctionEqual}).

		\begin{figure}[ht!]
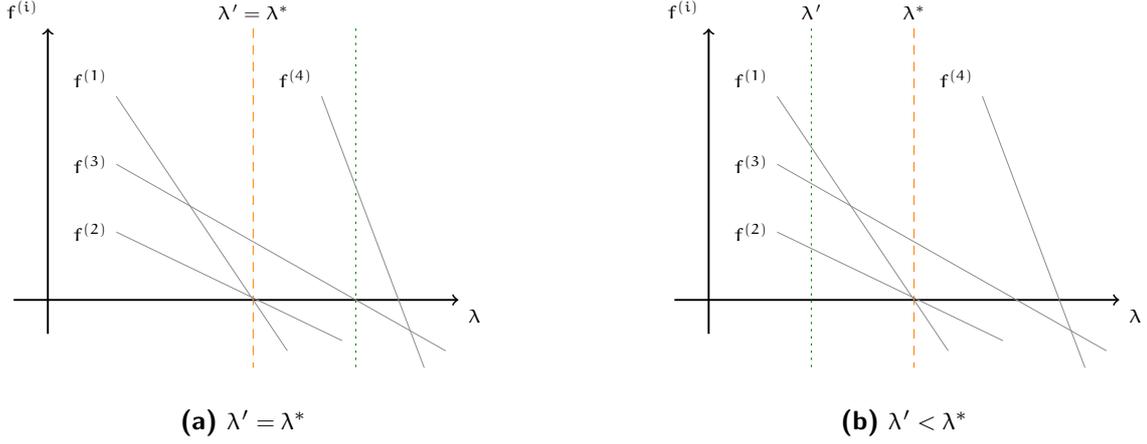

			\subfloat[][$\lambda' = \lambda^*$]{
				\centering
				\begin{minipage}[b][5.5cm][t]{0.5\textwidth}
					\centering
					\scaledtikz{0.9}{
						\draw[->,thick] (-0.5,0) -- (6,0) node[below right] {\scriptsize $\lambda$};
						\draw[->,thick] (0,-0.5) -- (0,4) node[above left] {\scriptsize $f^{(i)}$};

						\draw[dashed,orange] (3,4) -- (3,-1.0) node[above,pos=0,black] {\scriptsize $\lambda' = \lambda^*$};
						\draw[dotted,darkgreen] (4.5,4) -- (4.5,-1.0);

						\draw[halfgray] (1,3) -- (3.5,-0.75) node[above left,pos=0,black] {\scriptsize $f^{(1)}$};
						\draw[halfgray] (1,1) -- (4.3,-0.6) node[left,pos=0,black] {\scriptsize $f^{(2)}$};

						\draw[halfgray] (1,2) -- (5.8125,-0.75) node[left,pos=0,black] {\scriptsize $f^{(3)}$};
						\draw[halfgray] (4,3) -- (5.5,-1) node[above left,pos=0,black] {\scriptsize $f^{(4)}$};
					}
				\end{minipage}
				\label{fig:ParametricFunctionEqual}
			}
			\subfloat[][$\lambda' < \lambda^*$]{
				\centering
				\begin{minipage}[b][5.5cm][t]{0.5\textwidth}
					\centering
					\scaledtikz{0.9}{
						\draw[->,thick] (-0.5,0) -- (6,0) node[below right] {\scriptsize $\lambda$};
						\draw[->,thick] (0,-0.5) -- (0,4) node[above left] {\scriptsize $f^{(i)}$};

						\draw[dotted,darkgreen] (1.5,4) -- (1.5,-1.0) node[above,pos=0,black] {\scriptsize $\lambda'$};
						\draw[dashed,orange] (3,4) -- (3,-1.0) node[above,pos=0,black] {\scriptsize $\lambda^*$};

						\draw[halfgray] (1,3) -- (3.5,-0.75) node[above left,pos=0,black] {\scriptsize $f^{(1)}$};
						\draw[halfgray] (1,1) -- (4.3,-0.6) node[left,pos=0,black] {\scriptsize $f^{(2)}$};

						\draw[halfgray] (1,2) -- (5.8125,-0.75) node[left,pos=0,black] {\scriptsize $f^{(3)}$};
						\draw[halfgray] (4,3) -- (5.5,-1) node[above left,pos=0,black] {\scriptsize $f^{(4)}$};
					}
				\end{minipage}
				\label{fig:ParametricFunctionSmaller}
			}
			\caption[Illustration of the two cases that might occur during the simulation of the separation oracle]{Illustration of the two cases that may occur during the simulation of the separation oracle in case that the separation oracle did not return a certificate when evaluated for a candidate value~$\lambda'$.}
			\label{fig:ParametricFunction}
		\end{figure}

		Now suppose that $\lambda' < \lambda^*$ (cf. Figure~\ref{fig:ParametricFunctionSmaller}). In this case, for a sufficiently small but positive value of $\delta$, the separation oracle returns the \emph{same} answer when called with the cost vector~$d(\lambda' + \delta)$ as long as $\lambda' + \delta \leq \lambda^*$ (i.e., as long as $\delta$ is smaller than the distance between the dotted and the dashed line in Figure~\ref{fig:ParametricFunctionSmaller}). Consequently, in order to separate this case from the former case, it suffices to specify a value for $\delta$ that is smaller than the distance between any two roots of the functions that occur both in the instance and during the simulation of $\mathcal{A}$. We can then use a second call to the decision oracle in order to decide whether a candidate value~$\lambda'$ is smaller than or equal to the optimal value~$\lambda^*$.

		First note that the root of each function~$f^{(i)}$ is given by the rational number~$\frac{a^T x^{(i)}}{c^T x^{(i)}}$. Since the coefficients~$c_j$ are part of the instance~$I$ of the problem~\eqref{eqn:OriginalProblem} and since the values~$a_j = \sum_{i=1}^m y_i \cdot A_{ij}$ are generated within the framework of \citet{GargKoenemann}, the encoding length of both values is polynomial in the problem size. Similarly, as noted in Section~\ref{sec:Cones}, we can assume that the encoding lengths of all vectors~$x^{(i)}$ returned by the oracle are in $\Omega(n)$ and, since the oracle runs in polynomial time, polynomially bounded in the instance size. Consequently, there is some bound~$M_f$ with polynomial encoding length such that the root of each function~$f^{(i)}$ can be represented by a fraction~$\frac{p_i}{q_i}$ with $p_i,q_i \in \mathbb{N}$ and $q_i \leq M_f$.

		Now consider the root~$-\frac{a_0 - b_0}{a_1 - b_1}$ of some function~$g$ of the form $g(\lambda) \colonequals (a_0 - b_0) + \lambda \cdot (a_1 - b_1)$ that stems from a comparison of two linear parametric values of the forms~$a_0 + \lambda \cdot a_1$ and $b_0 + \lambda \cdot b_1$. Assume that we are in the $k$-th step of the simulation. Since the oracle algorithm is strongly combinatorial, the values~$a_0 + \lambda \cdot a_1$ and $b_0 + \lambda \cdot b_1$ result from one or more of the input values~$d_j \colonequals a_j - \lambda \cdot c_j$ (which are the only linear parametric values at the beginning of the simulation) as well as a sequence of up to $k$ additions or subtractions with other linear parametric values and multiplications with constants. Hence, since $k \in \OO(T_\mathcal{A})$ and $\mathcal{A}$ runs in (strongly) polynomial time, there is a bound~$M_g$ with polynomial encoding length such that the root~$-\frac{a_0 - b_0}{a_1 - b_1}$ of each such function~$g$ considered up to the $k$-th step of the simulation can be represented by a fraction of the form~$\frac{p}{q}$ with $p,q \in \mathbb{N}$ and $q \leq M_g$.

		Now let $\mu_1 = \frac{p_1}{q_1}$ and $\mu_2 = \frac{p_2}{q_2}$ with $\mu_1 \neq \mu_2$ denote the roots of two of the above functions of the form~$f^{(i)}$ or $g$. Since $q_1,q_2 \leq M_f \cdot M_g$, we get that
		\begin{align*}
			\left| \mu_1 - \mu_2 \right| = \left| \frac{p_1}{q_1} - \frac{p_2}{q_2} \right| = \left| \frac{p_1 \cdot q_2 - p_2 \cdot q_1}{q_1 \cdot q_2} \right| \geq \frac{1}{M_f^2 \cdot M_g^2} \equalscolon \mu
		\end{align*}
		Hence, choosing $\delta \colonequals \frac{\mu}{2}$, we are able to differentiate between the three cases~$D(\lambda) < 0$, $D(\lambda) = 0$, and $D(\lambda) > 0$. Moreover, by returning \emph{any\footnotemark} vector in $S$ in the case of $D(\lambda) > 0$ and returning the certificate in every other case, the separation oracle is extended into a sign oracle and the correctness follows from the proof of Lemma~\ref{lem:SignOracleMostViolatedMegiddo}. Note that the running time remains to be $\OO\left(N + T^2_\mathcal{A} \right)$ (as in the case of a sign oracle in Lemma~\ref{lem:SignOracleMostViolatedMegiddo}) since the encoding length of the number~$\delta$ is polynomially bounded and since the oracle algorithm is assumed to run in strongly polynomial time.
		\footnotetext{Actually, since we do not have direct access to the set~$S$, we need to obtain such a vector via an oracle access. However, by calling the oracle once more with a very large value for~$\lambda$ or by returning \emph{some} vector found before, we obtain a certificate in $S$, which we can return.}
	\end{proof}

	Lemma~\ref{lem:SeparationOracleMostViolatedMegiddo} now yields one of the main results of this paper:

	\begin{theorem}\label{thm:FPTASStronglySeparation}
		Given a strongly combinatorial and strongly polynomial-time sign oracle~$\mathcal{A}$ for $S$ running in $T_\mathcal{A}$~time, there is a strongly polynomial FPTAS for the problem~\eqref{eqn:OriginalProblem} running in $\OO\left(\frac{1}{\varepsilon^2} \cdot m\log m \cdot \left(N + T_\mathcal{A} \right) + T^2_\mathcal{A} \right)$ time. \qed
	\end{theorem}

	\section{Applications}
	\label{sec:Applications}

	In this section, we present several applications of the introduced framework. We will be able to derive new strongly polynomial-time FPTASs for several well-known network flow and packing problems and complement or even improve upon well-known results. All graphs considered in this section are assumed to be connected, such that the number of nodes~$n$ fulfills $n \in \mathcal{O}(m)$.

	\subsection{Budget-Constrained Maximum Flows}
	\label{sec:Applications:BCMFP}

	In the \emph{budget-constrained maximum flow problem}, the aim is to determine a flow with maximum value in an $s$-$t$-network that is additionally restricted by a \emph{budget-constraint} of the form $\sum_{e \in E} b_e \cdot x_e \leq B$ for non-negative integers~$b_e \in \mathbb{N}$ for each $e \in E$ a budget~$B \in \mathbb{N}$. The problem is known to be efficiently solvable by combinatorial algorithms, both in weakly polynomial-time \citep{AhujaConstrainedMaxFlow,CalicskanDoubleScaling,CalicskanAhuja,CalicskanCostScaling} and in strongly polynomial-time \citep{BudgetConstrainedMinCostFlows}. In the following, we present a strongly polynomial-time FPTAS for the problem, which is both much more simple and efficient than the exact strongly polynomial-time algorithm.

	In order to apply our framework, we need to show that each feasible solution is decomposable in some kind of basic component and that we are able to handle these components appropriately. Without loss of generality, since each budget-constrained maximum flow~$x$ is also a traditional $s$-$t$-flow and since flows on cycles do not contribute to the flow value, it holds that $x$ can be decomposed into $m' \leq m$ flows~$\o{x}^{(j)}$ on $s$-$t$-paths~$P_j$ such that $x = \sum_{j=1}^{m'} \o{x}^{(j)}$. Hence, if $x^{(l)}$ denotes the flow with unit flow value on some path~$P_l$ in the set of $s$-$t$-paths~$\{P_1,\ldots,P_k\}$, it holds that each (budget-constrained) maximum flow~$x$ is contained in the cone~$C$ that is generated by the vectors in the set~$S \colonequals \{x^{(l)} : l \in \{1,\ldots,k\}\}$. Consequently, we can formulate the budget-constrained maximum flow problem as follows:
	\begin{subequations}\label{eqn:CMFP}
	\begin{align}
		\max\ & \sum_{e \in E} c_e \cdot x_e \\
		\text{s.t.}\ & \sum_{e \in E} b_e \cdot x_e \leq B, \\
		& x_e \leq u_e && \text{for each } e \in E, \\
		& x \in C,
	\end{align}
	\end{subequations}
	where $c_e = 1$ if $e \in \delta^-(t)$, and $c_e = 0$, else. Note that the flow conservation constraints are now modeled by the containment in the cone~$C$, such that a packing problem over a polyhedral cone remains, i.e., a problem of the form~\eqref{eqn:OriginalProblem}.

	In the above formulation, it holds that $c^T x^{(l)} = \hat{c} \colonequals 1$ for each $x^{(l)} \in S$ since each $s$-$t$-path contributes equally to the value of the flow. Hence, we can apply Theorem~\ref{thm:FPTASMinimizing} if we can show that there is a minimizing oracle for $S$, i.e., that we can determine a vector~$x^{(l)}$ minimizing~$d^T x^{(l)}$ for a given cost vector~$d$. This simply reduces to the determination of a shortest $s$-$t$-path with respect to the edge lengths~$d$. Note that, since the vector~$a$ is always positive in each component according to Observation~\ref{obs:EverythingPositive} and since $\hat{c} = 1$, we need to search for a shortest path with non-negative edge lengths in $\SP(m,n) \in \OO(m + n\log n)$~time according to the proof of Theorem~\ref{thm:FPTASMinimizing}. Thus, we get that there is an FPTAS for the budget-constrained maximum flow problem running in $\OO\left(\frac{1}{\varepsilon^2} \cdot m\log m \cdot \SP(m,n) \right)$~time since the number~$N$ of non-zero entries in the constraint matrix in \eqref{eqn:CMFP} is bounded by $2m \in \OO(SP(m,n))$. Note that this running time is still obtained even if we add (a constant number of) different budget-constraints.

	We want to stress that our framework allows to stick to the commonly used edge-based formulation of the problem, in which there is a linear number of variables defining the flow on single edges. In contrast, one is required to use the path-based formulation of the problem when using the original framework of \citet{GargKoenemann}: The flow conservation constraints, which define the ``shape'' of a feasible flow, cannot be directly used in a formulation as a packing problem. These constraints, however, are now modeled by the containment in the cone~$C$. Moreover, note that the only ingredients that we used are that (1) each flow decomposes into flows on some type of basic components ($s$-$t$-paths) and (2) that we are able to handle these basic components efficiently, which allowed us to apply the framework.

	\subsection{Budget-Constrained Minimum Cost Flows}
	\label{sec:Applications:BCMCFP}

	In the \emph{budget-constrained minimum cost flow problem}, the aim is to determine a minimum cost flow subject to a budget constraint of the form~$\sum_{e \in E} b_e \cdot x_e \leq B$, similarly to the budget-constrained maximum flow problem that was studied above. The problem is known to be efficiently solvable in weakly and strongly polynomial-time \citep{BudgetConstrainedComplexityApproximability,BudgetConstrainedMinCostFlows}. In \citep{BudgetConstrainedComplexityApproximability}, a strongly polynomial-time FPTAS was presented for the budget-constrained minimum cost flow problem, which runs in
	\[ \OO\left(\frac{1}{\varepsilon^2} \cdot m \log m \cdot (nm \log m \log\log m + n^3 \log n + nm \log^2 n \log\log n) \right) \]
	time and which uses similar ideas as the ones presented above. In the following, we improve upon this result.

	When considering the (equivalent) circulation based version of the problem in which there are no demands and flow conservation holds at each node, it is easy to see that each optimal flow can be decomposed into flows on simple cycles. Hence, we can restrict our considerations to flows that are contained in the cone~$C$ that is spanned by flows on simple cycles with unit flow value. The result of Theorem~\ref{thm:FPTASMinimizing} cannot be applied to this problem for two reasons: On the one hand, since we are dealing with arbitrary costs, it clearly does no longer hold that $c^T x^{(l)}$ is constant among all flows on cycles with unit flow value. On the other hand, any minimizing oracle would be required to return a vector~$x^{(l)}$ that minimizes $d^T x^{(l)}$ for a given cost vector~$d$, so it would be necessary to find a most negative cycle~$C^*$ in the underlying graph. However, this problem is known to be $\NP$-complete in general since finding a most negative simple cycle in a graph with edge costs~$d_e = -1$ for each $e \in E$ is equivalent to deciding if the graph contains a Hamiltonian cycle (cf. \citet{GareyJohnson}). Nevertheless, we are able to determine a cycle~$C$ with the same \emph{sign} as the most negative cycle~$C^*$ efficiently by computing a minimum mean cycle in $\OO(nm)$~time (cf. \citep{KarpMinimumMeanCycle}). Hence, we can apply both Theorem~\ref{thm:FPTASWeakly} and Theorem~\ref{thm:FPTASStronglySign} to the budget-constrained minimum cost flow problem in order to obtain a weakly polynomial-time FPTAS running in
	\begin{align*}
		\OO\left(nm \cdot \left(\frac{1}{\varepsilon^2} \cdot m\log m + \log\log M - \log \frac{1}{\varepsilon} - \log m - \log\log m \right) \right)
	\end{align*}
	time and, since the minimum mean cycle algorithm of \citet{KarpMinimumMeanCycle} is both strongly polynomial and strongly combinatorial, a strongly polynomial-time FPTAS with a time bound of
	\begin{align*}
		\OO\left(\frac{1}{\varepsilon^2} \cdot m\log m \cdot nm + (nm)^2 \right) = \OO\left(nm \cdot \left(\frac{1}{\varepsilon^2} \cdot m\log m + nm \right) \right).
	\end{align*}

	The latter running time can be improved by making use of the following observation: As it was shown in Lemma~\ref{lem:SignOracleMostViolatedMegiddo}, the sign oracle is incorporated into Megiddo's parametric search in order to determine a minimizer of
    \begin{align}
		\min_{\genfrac{}{}{0pt}{}{l \in \{1,\ldots,k\}}{c^T x^{(l)} > 0}} \dfrac{a^T x^{(l)}}{c^T x^{(l)}} \tag{\ref{eqn:MostViolatedSubproblem}}
	\end{align}
    for a positive cost vector~$a$ and a vector~$c$. In the case of the budget-constrained minimum cost flow problem, this reduces to the determination of a \emph{minimum ratio cycle}~$C$. \citet{MegiddoParallel} derived an algorithm that determines a minimum ratio cycle in a simple graph in $\OO(n^3 \log n + n m \log^2n \log\log n)$ time by making use of a parallel algorithm for the all-pairs shortest path problem in combination with \citeauthor{KarpMinimumMeanCycle}'s minimum mean cycle algorithm \citep{KarpMinimumMeanCycle} as a negative cycle detector in his parametric search. This running time was later improved by \citet{ColeMegiddo} to $\OO(n^3 \log n + n m \log^2n)$. Hence, the strongly polynomial FPTAS can be improved to run in
    \begin{align*}
		\OO\left(\frac{1}{\varepsilon^2} \cdot m\log m \cdot nm + n^3 \log n + n m \log^2n \right) = \OO\left(\frac{1}{\varepsilon^2} \cdot n m^2 \log n \right)
	\end{align*}
	time on simple graphs. In the case of multigraphs, one can use a technique introduced in \citep{BudgetConstrainedComplexityApproximability} in order to transform the graph into an equivalent simple graph in $\OO(nm \log m \log\log m)$~time before applying Cole's minimum ratio cycle algorithm, yielding an FPTAS running in
	\begin{align*}
		\OO\left(\frac{1}{\varepsilon^2} \cdot m\log m \cdot nm + nm \log m \log\log m + n^3 \log n + n m \log^2n \right) = \OO\left(\frac{1}{\varepsilon^2} \cdot n m^2 \log m \right)
	\end{align*}
	time. Hence, in both cases, the strongly polynomial-time FPTAS dominates the FPTASs introduced above. The claimed running time holds even if we add up to $\OO(n)$~different budget constraints to the problem.

    \subsection{Budget-Constrained Minimum Cost Generalized Flows}
    \label{sec:Applications:MCGFP}

    The \emph{generalized minimum cost flow problem} is an extension of the minimum cost flow problem, in which each edge~$e \in E$ is denoted with an additional \emph{gain factor}~$\gamma_e$. The flow that enters some edge~$e$ is multiplied by $\gamma_e$ as soon as it leaves the edge (cf. \citep{WayneGeneralizedFlows}). In the \emph{budget-constrained minimum cost generalized flow problem}, the flow is additionally restricted by a budget-constraint of the form $\sum_{e \in E} b_e \cdot x_e \leq B$ as above.

    The traditional minimum cost generalized flow problem (without an additional budget constraint) is known to be solvable by combinatorial algorithms in weakly polynomial-time \citep{WayneMinCostGeneralizedFlow}. Moreover, there is a strongly polynomial-time FPTAS running in $\OOT\left( \log \frac{1}{\varepsilon} \cdot n^2 m^2 \right)$ time presented by \citet{WayneMinCostGeneralizedFlow}. However, this algorithm makes use of the inherent structure of traditional generalized flows and cannot be extended to the budget-constrained case without further ado. Earlier, \citet{OldhamGeneralizedFlowApproximations} presented an FPTAS for the related generalized minimum cost maximum flow problem with non-negative edge costs with a weakly polynomial running time of $\OOT\left(\frac{1}{\varepsilon^2} \cdot \log \frac{1}{\varepsilon} \cdot n^2 m^2 \cdot \log mCU \right)$, which, as well, cannot be easily extended to the budget-constrained case. Another weakly polynomial-time FPTAS for this problem running%
    \footnote{$M$ denotes the largest absolute value of each number given in the problem instance, assuming gain factor are given as ratios of integers.}\xspace%
    in $\OOT\left( \frac{1}{\varepsilon^2} \cdot n m^2 \cdot (\log \frac{1}{\varepsilon} + \log\log M) \right)$ time was presented by \citet{WayneApproximation}, which is also based on the procedure of \citet{GargKoenemann} and which can be extended to the budget-constrained version of the problem. Using our framework, we present two much simpler FPTASs that work for the generalized minimum cost flow with arbitrary edge costs and that complement the above ones by achieving better time complexities in specific cases.

    Again, we consider the circulation based version of the problem in which the excess is zero at each node~$v \in V$. As it was shown in \citep{WayneMinCostGeneralizedFlow}, every such generalized circulation~$x$ can be decomposed into at most $m$ flows on \emph{unit-gain cycles} and \emph{bicycles}, i.e., flows on cycles~$C$ with $\prod_{e \in C} \gamma_e = 1$ and flows on pairs of cycles~$(C_1,C_2)$ with $\prod_{e \in C_1} \gamma_e < 1$ and $\prod_{e \in C_2} \gamma_e > 1$ that are connected by a path, respectively. Hence, every generalized circulation lies in the cone~$C$ that is generated by flows on such unit-gain cycles and bicycles:
    \begin{subequations}\label{eqn:GMCFP}
	\begin{align}
		\max\ & \sum_{e \in E} c_e \cdot x_e \\
		\text{s.t.}\ & \sum_{e \in E} b_e \cdot x_e \leq B, \\
		& x_e \leq u_e && \text{for each } e \in E, \\
		& x \in C.
	\end{align}
	\end{subequations}
    Note that this formulation does not differ from the models in the previous applications. Instead, the ``structural complexity'' of the problem that comes with the introduction of gain factors is modeled by the containment in the cone~$C$. We are done if we are able to find a separation oracle for the set that generates this cone. \citet{WayneMinCostGeneralizedFlow} shows that there is a unit-gain cycle or bicycle with negative costs in a given network if and only if a specialized system with two variables per inequality (2VPI) is infeasible. Among others, \citet{CohenTVPI} present a procedure that checks the feasibility of such a system and, in case that it is infeasible, provides a ``certificate of infeasibility'', which corresponds to a negative cost unit-gain cycle/bicycle \citep{WayneMinCostGeneralizedFlow}. This procedure runs in $\OOT(n)$~time on $\OO(nm)$ processors. Hence, when used as a separation oracle, we are able to apply Theorem~\ref{thm:FPTASWeakly}. This yields an FPTAS running in
    \begin{align*}
    	\OOT\left(n^2 m \cdot \left(\frac{1}{\varepsilon^2} \cdot m + \log\log M' - \log \frac{1}{\varepsilon} \right) \right)
    \end{align*}
    time, where $M'$ is an upper bound on the absolute costs~$c_e$, fees~$b_e$, and capacities~$u_e$ of the edges~$e \in E$ -- independent of the numbers involved to represent the gain factors. Moreover, since the separation oracle is both strongly polynomial and strongly combinatorial \citep{WayneMinCostGeneralizedFlow}, we can apply Theorem~\ref{thm:FPTASStronglySign} in order to obtain a strongly polynomial-time FPTAS. Using parallelization techniques that are common when using Megiddo's parametric search \citep{MegiddoParallel}, the time that is necessary to find an initial most violated dual constraint using Lemma~\ref{lem:SeparationOracleMostViolatedMegiddo} can be improved from $\OOT((nm)^2)$ to $\OOT(n \cdot (nm + nm \log (nm) + \log(nm) \cdot (n^2m))) = \OOT(n^3 m)$. This yields an FPTAS with a strongly polynomial running time in
    \begin{align*}
    	\OOT\left(\frac{1}{\varepsilon^2} \cdot m \cdot n^2 m + n^3 m \right) = \OOT\left(\frac{1}{\varepsilon^2} \cdot n^2 m^2 \right).
    \end{align*}
    This algorithm embodies the first strongly polynomial-time FPTAS for the budget-constrained generalized minimum cost flow problem and improves upon the running time of the weakly polynomial-time FPTAS. Moreover, this FPTAS outperforms both the algorithm of \citet{OldhamGeneralizedFlowApproximations} and, for large values of $M$ or small values of $\varepsilon$, the algorithm of \citet{WayneApproximation}.

    \subsection{Maximum Flows in Generalized Processing Networks}
    \label{sec:Applications:MFGPN}

    \emph{Generalized processing networks} extend traditional networks by a second kind of capacities, so called \emph{dynamic capacities}, that depend on the flow itself. More precisely, the flow on each edge~$e=(v,w) \in E$ is additionally constrained to be at most $\alpha_e \cdot \sum_{e' \in \delta^+(v)} x_{e'}$ for some edge-dependent constant~$\alpha_e \in (0,1]$, i.e., the flow on $e$ may only make up a specific fraction~$\alpha_e$ of the total flow leaving the starting node~$v$ of $e$. This extension allows to model manufacturing and distillation processes, in particular (cf. \citep{MaxFlowsInGeneralizedProcessingNetworks}).

	Similar to $s$-$t$-paths, the ``basic component'' in the field of generalized processing networks is the notion of so-called \emph{basic flow distribution schemes}. For each node~$v \in V$ with $\delta^+(v) \neq \emptyset$, such a basic flow distribution scheme~$\beta$ is a function that assigns a value in $[0,\alpha_e]$ to each edge~$e \in \delta^+(v)$ such that $\sum_{e \in \delta^+(v)} \beta_e = 1$ and at most one edge~$e \in \delta^+(v)$ fulfills $\beta_e \in (0,\alpha_e)$. Intuitively, a basic flow distribution scheme describes how flow can be distributed to the outgoing edges at each node without violating any dynamic capacity constraint.

	In \citep{MaxFlowsInGeneralizedProcessingNetworks}, the authors show that each flow in a generalized processing network can be decomposed into at most $2m$ flows on basic flow distribution schemes. Hence, we can conclude that each \emph{maximum flow in a generalized processing network} is contained in the cone~$C$ that is generated by unit-flows on basic flow distribution schemes. Moreover, for the problem on acyclic graphs and for a given cost vector~$d$, we can determine a basic flow distribution scheme~$\beta$ that allows a unit-flow~$x$ with minimum costs~$d(x) \colonequals \sum_{e \in E} d_e \cdot x_e$ in linear time~$\OO(m)$ (cf. \citep{MaxFlowsInGeneralizedProcessingNetworks}). By using Theorem~\ref{thm:FPTASMinimizing}, we get an FPTAS for the maximum flow problem in generalized processing networks with a strongly polynomial running-time of $\OO(\frac{1}{\varepsilon^2} \cdot m^2 \log m)$. This result is in particular interesting since it is unknown whether an \emph{exact} solution can be determined in strongly polynomial time since the problem is at least as hard to solve as any linear fractional packing problem (cf. \citep{MaxFlowsInGeneralizedProcessingNetworks} for further details).

	\subsection{Minimum Cost Flows in Generalized Processing Networks}
	\label{sec:Applications:MCFGPN}

	Similar to the previous problem, each \emph{minimum cost flow in a generalized processing network} is contained in the cone that is generated by flows with unit flow value on basic flow distribution schemes. On acyclic graphs, we have the same minimizing oracle as described above. Since the costs~$c_e$ are now arbitrary, we can no longer apply Theorem~\ref{thm:FPTASMinimizing}. Nevertheless, since each minimizing oracle induces a sign oracle, we are able to apply Theorem~\ref{thm:FPTASStronglySign}, which yields an FPTAS for the problem running in strongly polynomial-time
	\begin{align*}
		\OO\left(\frac{1}{\varepsilon^2} \cdot m \log m \cdot m + m^2 \right) = \OO\left(\frac{1}{\varepsilon^2} \cdot m^2 \log m \right).
	\end{align*}
	This matches the time complexity of the maximum flow variant of the problem described in Section~\ref{sec:Applications:MFGPN}.

	\subsection{Maximum Concurrent Flow Problem}
	\label{sec:Applications:MCCFP}

	The \emph{maximum concurrent flow problem} is a variant of the maximum multicommodity flow problem, in which a demand~$d_j$ is given for each commodity~$j$ with source-sink-pair~$(s_j,t_j) \in V \times V$. The task is to determine the maximum value of $\lambda$ such that a fraction~$\lambda$ of all demands is satisfied without violating any edge capacity. While several FPTASs emerged for this problem, the best time bound at present is given by $\OOT\left(\frac{1}{\varepsilon^2} \cdot (m^2 + kn) \right)$ due to \citet{KarakostasConcurrentFlowProblem}, where $k \in \OO(n^2)$ denotes the number of commodities.

	The problem can be approximated efficiently with our framework by using the following novel approach: In order to improve the objective function value by one unit, we need to send $d_j$~units of each commodity. Hence, each concurrent flow decomposes into basic components of the following type: A set of flows on $k$~paths, containing a flow with value~$d_j$ on an $(s_j,t_j)$~path for each commodity~$j$. For a given (positive) cost vector, a basic component with minimum costs can be found by determining a shortest path between each commodity. Since \citets{Dijkstra} algorithm computes the shortest paths from one node to \emph{every} other node, we only need to apply it $\min\{k,n\}$~times (once for each of the distinct sources of all commodities), which yields a minimizing oracle running in $\OO(\min\{k,n\} \cdot (m + n\log n))$~time and an FPTAS running in $\OOT\left(\frac{1}{\varepsilon^2} \cdot m^2 \cdot \min\{k,n\}\right)$~time according to Theorem~\ref{thm:FPTASMinimizing}. This algorithm has a worse time complexity than the one of \citet{KarakostasConcurrentFlowProblem}. Nevertheless, the application of the presented framework is much simpler than the algorithm given in \citep{KarakostasConcurrentFlowProblem} (and even matches its time complexity in the case of sparse graphs with a large number of commodities) and inherently allows the incorporation of additional budget-constraints.

	\subsection{Maximum Weighted Multicommodity flow Problem}
	\label{sec:Applications:MWMCFP}

	The \emph{maximum weighted multicommodity flow problem} is a generalization of the maximum multicommodity flow problem, in which a positive weight~$c_j$ is denoted with each commodity and the aim is to maximize the weighted flow value. The problem is known to be solvable in $\OOT\left( \frac{1}{\varepsilon^2} \cdot m^2 \min\{\log C, k\} \right)$ time as shown by \citet{FleischerMulticommodity}, where $C$ denotes the largest ratio of any two weights of commodities.

	Similar to the multicommodity flow problem, each feasible flow decomposes into flows with unit flow value on single $(s_j,t_j)$~paths. Moreover, the determination of such a path with minimal costs reduces to $\min\{k,n\}$ shortest path computations with possibly negative costs, similar to the maximum concurrent flow problem considered above. Using similar ideas as in the case of the budget-constrained minimum cost flow problem (Section~\ref{sec:Applications:BCMCFP}), this would yield an FPTAS with a running time in $\OOT\left( \frac{1}{\varepsilon^2} \cdot \min\{n,k\}\cdot nm + \min\{n,k\} \cdot n^3 \right)$.

	This running time can be improved as follows: As above, we can consider the cone~$C$ to be spanned by flows on $(s_j,t_j)$~paths for each commodity, but where each flow between any $(s_j,t_j)$-pair now has flow value~$\frac{1}{c_j}$. Each vector in the ground set~$S$ then has uniform costs. In order to apply Theorem~\ref{thm:FPTASMinimizing}, we need to be able to determine a cost-minimal vector with respect to a given positive cost vector~$d$. One straight-forward way to obtain such a minimizing oracle is to compute a shortest path for each commodity~$j$ using the edge lengths~$\frac{d_e}{c_j}$ for each $e \in E$ and to choose a shortest path among all commodities. This would result in a $\OOT\left(\frac{1}{\varepsilon^2} \cdot m^2 \cdot k \right)$~time FPTAS, similar to the previous application. However, it suffices to compute only $\min\{n,k\}$~shortest paths per iteration, which can be seen as follows: For each node~$s$ out of the set of the $\min\{k,n\}$ distinct source nodes, we perform two steps: We first compute the shortest path distance to every other node using \citets{Dijkstra} algorithm. Afterwards, for each node that corresponds to the sink~$t_j$ of a commodity~$j$ with source~$s_j=s$, we multiply the distance from $s_j$ to $t_j$ by $\frac{1}{c_j}$. By repeating this procedure for each source and keeping track of the overall shortest path, we obtain a minimizing oracle. This yields an FPTAS running in $\OOT\left( \frac{1}{\varepsilon^2} \cdot m^2 \cdot \min\{n,k\} \right)$~time, which complements the result of \citet{FleischerMulticommodity}. This example shows that more sophisticated definitions of the ground set~$S$ and the cone~$C$ may improve the running time of the procedure.

	Finally, using this approach, we can even further improve the algorithm to obtain a time bound of $\OOT\left( \frac{1}{\varepsilon^2} \cdot m^2 \right)$ using an idea that was applied by \citet{FleischerMulticommodity} to the (unweighted) multicommodity flow problem: For an initially tight lower bound~$\u{L}$ on the length of a shortest path for any commodity (which can be computed in $\OOT(\min\{n,k\} \cdot m)$~time as above at the beginning), we can stick to one commodity~$j$ in each iteration of the overall procedure and compute a single shortest path from the source~$s_j$ to the sink~$t_j$. Once the length of this shortest path multiplied by $\frac{1}{c_j}$ becomes as large as $(1 + \varepsilon) \cdot \u{L}$, we go on to the next commodity and continue the procedure. After each commodity was considered, we update $\u{L}$ to $(1 + \varepsilon)\cdot \u{L}$ and continue with the first commodity. Following the lines of \citet{FleischerMulticommodity}, this yields an FPTAS running in $\OOT\left( \frac{1}{\varepsilon^2} \cdot (m^2 + km) \right)$~time as there are $\OOT\left( \frac{1}{\varepsilon^2} \cdot k \right)$ shortest path computations that lead to a change of the commodity. However, since \citets{Dijkstra} algorithm computes the distance to every other node, we only need to consider $\min\{k,n\}$ different nodes by grouping commodities with the same source as above, which reduces the running time to $\OOT\left( \frac{1}{\varepsilon^2} \cdot m^2 \right)$ (see \citep{FleischerMulticommodity} for details on the algorithm). Although \citet{FleischerMulticommodity} both considered this technique and introduced the maximum weighted multicommodity flow problem, she refrained from applying this procedure to the problem.

	\subsection{Maximum Spanning Tree Packing Problem}
	\label{sec:Applications:SpanningTrees}

	In the \emph{maximum spanning tree packing problem}, one is given an undirected graph~$G=(V,E)$ with positive edge capacities~$u_e$. Let $\mathcal{T}$ denote the set of all spanning trees in $G$. The aim is to find a solution to the problem
	\begin{align*}
		\max\ & \sum_{T \in \mathcal{T}} x_T \\
		\text{s.t.}\ & \sum_{T \in \mathcal{T}: e \in T} x_T \leq u_e && \forall\ e \in E, \\
		& x_T \geq 0 && \forall T \in \mathcal{T},
	\end{align*}
	i.e., one seeks to pack as many spanning trees as possible (in the fractional sense) without violating any edge capacity. While the problem was investigated in a large number of publications, the fastest (exact) algorithm for the problem is due to \citet{GabowPackingTrees} and runs in $\OO\left( n^3 m \log \frac{n^2}{m} \right)$~time.

	Let $S$ denote the set of incidence vectors~$\chi_T$ of spanning trees~$T \in \mathcal{T}$, where $(\chi_T)_e = 1$ if $e \in T$ and $(\chi_T)_e = 0$ else. Since each spanning tree contains exactly~$n-1$ edges, the problem can be stated in an equivalent edge-based fashion as follows:
	\begin{align*}
		\max\ & \frac{1}{n-1} \cdot \sum_{e \in E} x_e \\
		\text{s.t.}\ & x_e \leq u_e && \forall\ e \in E, \\
		& x \in C.
	\end{align*}
	In order to apply Theorem~\ref{thm:FPTASMinimizing} (which is eligible since each spanning tree contributes equally to the objective function value), we need a minimizing oracle for the set~$S$. However, this simply reduces to the determination of a minimum spanning tree, which can be done in $\mathcal{O}(m \cdot \alpha(m,n))$~time, where $\alpha(m,n)$ denotes the inverse Ackermann function (cf. \citep{ChazelleMinimumSpanningTree}). This yields a strongly polynomial-time FPTAS for the problem running in $\OO\left( \frac{1}{\varepsilon^2} \cdot m^2 \log m \cdot \alpha(m,n)\right)$~time.

	Our framework also applies to a \emph{weighted} version of the problem: Assume that each edge is labeled with an additional cost~$c_e$ (with arbitrary sign) and assume that the weight~$c(T)$ of each spanning tree~$T \in \mathcal{T}$ is defined to be the sum of the weights of its edges, i.e., $c(T) \colonequals \sum_{e \in E} c_e$. The aim is then to maximize the objective function~$\sum_{T \in \mathcal{T}} c(T) \cdot x_T$. As above, we can stick to an equivalent edge-based formulation using the objective function~$\frac{1}{n-1} \cdot \sum_{e \in E} c_e \cdot x_e$. The minimum spanning tree algorithm can then be used as a sign oracle, which allows us to apply Theorem~\ref{thm:FPTASStronglySign} to the problem. This yields an FPTAS for the maximum weighted spanning tree packing problem running in strongly polynomial time
	\[ \OO\left( \frac{1}{\varepsilon^2} \cdot m^2 \log m \cdot \alpha(m,n) + m^2 \cdot \alpha^2(m,n) \right) = \OO\left( \frac{1}{\varepsilon^2} \cdot m^2 \log m \cdot \alpha(m,n) \right). \]
	To the best of our knowledge, this is the first combinatorial approximation algorithm for this problem.

	\subsection{Maximum Matroid Base Packing Problem}
	\label{sec:Applications:Matroids}

	Having a closer look at the results of Section~\ref{sec:Applications:SpanningTrees}, one might expect that they can be generalized to matroids: As spanning trees form the bases of graphic matroids, the presented ideas suggest that the framework can also be applied to packing problems over general matroids. In the \emph{maximum matroid base packing problem}, a matroid~$M(S,\mathcal{I})$ with ground set~$S \colonequals \{1,\ldots,m\}$ and independent sets in $\mathcal{I}$ is given as well as a positive capacity~$u_i \in \mathbb{N}_{> 0}$ for each $i \in S$. For $r$ to be the rank function of $M$, let $\mathcal{B} \subset \mathcal{I}$ denote the set of bases such that $I \in \mathcal{B}$ if and only if $r(I) = r(S)$. The aim of the problem is to pack as many bases of $M$ as possible (in the fractional sense) without violating any capacity constraints:
	\begin{align*}
		\max\ & \sum_{I \in \mathcal{B}} x_I \\
		\text{s.t.}\ & \sum_{I \in \mathcal{B}: i \in I} x_I \leq u_i && \forall\ i \in S, \\
		& x_I \geq 0 && \forall I \in \mathcal{B}.
	\end{align*}
	As it is common when dealing with matroids, we assume that the matroid is described by an \emph{independence testing oracle}, which checks if some set~$S' \subseteq S$ is independent in $M$ (cf. \citep{Schrijver2}). Let $F(m)$ denote the running time of this oracle. As it is shown in \citep{Schrijver2}, the problem can be solved in $\OO(m^7 \cdot F(m))$~time using a result derived by \citet{CunninghamMatroid}.

	As it was the case in the maximum spanning tree packing problem in Section~\ref{sec:Applications:SpanningTrees}, the problem can be formulated in an equivalent element-based fashion as follows:
	\begin{align*}
		\max\ & \frac{1}{r(S)} \cdot \sum_{i \in S} x_i \\
		\text{s.t.}\ & x_i \leq u_i && \forall\ i \in S, \\
		& x \in C,
	\end{align*}
	where the cone~$C$ is spanned by the incidence vectors of bases in $\mathcal{B}$. In order to apply our framework, we need to be able to handle these bases efficiently. However, as we are dealing with matroids, we can find a cost-minimal basis~$I \in \mathcal{B}$ of $M$ with respect to a given cost vector~$d$ just by applying the Greedy algorithm (cf. \citep{Schrijver2}): Starting with $I = \emptyset$, we sort the elements by their costs and iteratively add each element in the sorted sequence unless the independence test fails. This yields a minimizing oracle for $\mathcal{B}$ running in $\OO(m \cdot F(m) + m\log m)$~time. Hence, we immediately get an FPTAS for the maximum matroid base packing problem running in strongly polynomial time~$\OO\left(\frac{1}{\varepsilon^2} \cdot m^2 \log m \cdot (F(m) + \log m) \right)$ according to Theorem~\ref{thm:FPTASMinimizing}.

	As it was the case in Section~\ref{sec:Applications:SpanningTrees}, we can also extend our results to a \emph{weighted} version of the problem: Assume we are additionally given costs~$c_i \in \mathbb{Z}$ and want to maximize $\sum_{I \in \mathcal{B}} c(I) \cdot x_I$, where $c(I) \colonequals \sum_{i \in I} c_i$. Equivalently, we can also maximize $\frac{1}{r(S)} \cdot \sum_{i \in S} c_i \cdot x_i$ in the element-based formulation of the problem. Using the above minimizing oracle as a sign oracle, we can apply both Theorem~\ref{thm:FPTASWeakly} and, in case that the independence testing oracle is strongly polynomial and strongly combinatorial, Theorem~\ref{thm:FPTASStronglySign} to the problem. This yields two FPTASs for the problem running in
	\[ \OO\left( (m \cdot F(m) + m \log m) \cdot \left(\frac{1}{\varepsilon^2} \cdot m \log m + \log\log M - \log\frac{1}{\varepsilon} - \log m - \log\log m \right) \right) \]
	and
	\[ \OO\left( \frac{1}{\varepsilon^2} \cdot m^2 \log m \cdot (F(m) + \log m) + (m \cdot F(m) + m \log m)^2 \right) \]
	time, respectively. To the best of our knowledge, no other polynomial-time algorithm is known for this problem.

	\section{Conclusion}

	We investigated an extension of the fractional packing framework by \citet{GargKoenemann} that generalizes their results to fractional packing problems over polyhedral cones. By combining a large diversity of known techniques, we derived a framework that can be easily adopted to a large class of network flow and packing problems. This framework may in particular be applicable even if the cone has an exponential-sized representation as it only relies on a strongly polynomial number of oracle calls in order to gather information about the cone. In many cases, its application allows the derivation of approximation algorithms that are either the first ones with a strongly polynomial running time or the first combinatorial ones at all. For a large variety of applications, we were even able to complement or improve existing results.

	The presented paper raises several questions for future research. On the one hand, we believe that our results can be applied to a much larger set of problems and can be used to obtain combinatorial FPTASs for complex problems without much effort. It may also be possible that the results continue to hold for even weaker kinds of oracles. On the other hand, as our framework is based on the one of \citet{GargKoenemann} in its core, all of the derived approximation algorithms have a running time in $\Omega(\frac{1}{\varepsilon^2} \cdot m\log m \cdot n)$ and, in particular, have a quadratic dependency on $\frac{1}{\varepsilon}$. It may be possible to achieve a subquadratic dependency on $\frac{1}{\varepsilon}$ by relying on other approaches such as the one of \citet{BienstockIyengarFractionalPacking}. Nevertheless, it seems that this trade comes with a worse dependence on other parameters, a worse practical performance, or a worse generality of the presented results.

	\bibliographystyle{plainnat}
	\bibliography{../../../literature.bib}
\end{document}